\pdfoutput=1
\newif\ifarxiv
\arxivtrue 
\ifarxiv
	\documentclass[journal,onecolumn,12pt,twoside]{IEEEtranTCOM}
\else
	\documentclass[journal,draftclsnofoot,onecolumn,12pt,twoside]{IEEEtranTCOM}
\fi
\usepackage{etex}

\normalsize

\usepackage{cite}
\usepackage{setspace}

\usepackage{graphicx}
\usepackage{color}
\usepackage{pst-sigsys}
\ifCLASSINFOpdf
\else
\fi

\usepackage[cmex10]{amsmath}
\usepackage{amsmath,dsfont}
\usepackage{verbatim}
\usepackage{amssymb}
\usepackage{amsthm}
\usepackage{nccmath}
\usepackage{enumerate}
  \usepackage{pgfplots}
  \usepackage[normalem]{ulem}
  \usepackage{cases}
  \usepackage{bbm}  
  \usepackage{mathtools}
  \usepackage{hyperref}
  \usepackage{cleveref}
\pgfplotsset{plot coordinates/math parser=false}
\renewcommand{\descriptionlabel}[1]{%
  \hspace\labelsep \upshape{ \textit{\underline{#1:}}}%
}

\newtheorem{mydef}{Definition}

\newtheorem{theorem}{Theorem}

\DeclareMathOperator*{\argmin}{arg\,min}

\newcommand{\wi}{\ensuremath{w_{i}^{*}}}

\newcommand{\trans}[2]{\ensuremath{p_{#1#2}}}
\newcommand{\transM}{\ensuremath{P}}

\newcommand{\myState}{\ensuremath{\Omega}}

\newcommand{\chainsetarg}[1]{\ensuremath{\mathcal{C}(#1)}}
\newcommand{\chainsett}{\ensuremath{\chainsetarg{t}}}

\newcommand{\xsetarg}[1]{\ensuremath{\mathcal{T}(#1)}}
\newcommand{\xsett}{\ensuremath{\xsetarg{t}}}

\newcommand{\zall}{\ensuremath{\mathbf{Z}(t)}}

\newcommand{\myrew}{\ensuremath{\rho}}
\newcommand{\Mstar}{\ensuremath{M^{*}}}

\newcommand{\rhosub}[1]{\ensuremath{\myrew_{#1}}}

\newcommand{\rhoj}{\ensuremath{\rhosub{j}}}
\newcommand{\rhok}{\ensuremath{\rhosub{k}}}
\newcommand{\rhoqi}{\ensuremath{\rhosub{Q_{i}}}}
\newcommand{\rhoqj}{\ensuremath{\rhosub{Q_{j}}}}
\newcommand{\rhoqk}{\ensuremath{\rhosub{Q_{k}}}}
\newcommand{\rhoc}{\ensuremath{\rhosub{E}}}
\newcommand{\rhoqstar}{\ensuremath{\rhosub{\Qstar}}}
\newcommand{\Qstar}{\ensuremath{Q^{*}}}

\newcommand{\rhocarg}[2]{\ensuremath{\rhoc(#1, #2)}}

\newcommand{\orhoe}{\ensuremath{\bar{\myrew}_{E}}}
\newcommand{\rholcarg}[2]{\ensuremath{\orhoe(#1, #2)}}

\newcommand{\orhoj}{\ensuremath{\bar{\myrew}_{j}}}
\newcommand{\orhok}{\ensuremath{\bar{\myrew}_{k}}}
\newcommand{\orhoqi}{\ensuremath{\bar{\myrew}_{Q_{i}}}}
\newcommand{\orhoqj}{\ensuremath{\bar{\myrew}_{Q_{j}}}}
\newcommand{\orhoqk}{\ensuremath{\bar{\myrew}_{Q_{k}}}}
\newcommand{\orhoqstar}{\ensuremath{\bar{\myrew}_{\Qstar}}}

\newcommand{\myRew}{\ensuremath{R}}
\newcommand{\oRhojl}{\ensuremath{\bar{\myRew}_{u, l}}}
\newcommand{\oRhoel}{\ensuremath{\bar{\myRew}_{E, l}}}

\newlength\figureheight 
\newlength\figurewidth 

\newtheorem{myLemma}{Lemma}

\newtheorem{fact}{Fact}
\newcommand{\Q}[1]{\ensuremath{Q_{\{#1\}}}}
\newcommand{\tQ}[1]{\ensuremath{\tilde{Q}_{\{#1\}}}}
\newcommand{\tQo}[1]{\ensuremath{\tilde{Q}_{#1}}}
\newcommand{\Qp}[3]{\ensuremath{Q_{#1}^{+}(#2,#3)}}
\newcommand{\myd}[2]{\ensuremath{\delta_{#1}^{#2}}}
\newcommand{\bT}[1]{\ensuremath{\bar{T}_{#1}}}
\newcommand{\bTs}[1]{\ensuremath{\bar{T}_{#1}^{*}}}
\newcommand{\dvec}{\ensuremath{(d_1, d_2, d_3)}}

\newcommand{\markovcor}{4}
\newcommand{\markovthm}{4}
\newcommand{\lemOPSatisfiesCond}{11}
\newcommand{\lemCondSatisfiesOP}{12}
\newcommand{\lemOPunique}{13}

\newcounter{cnt}
\newcounter{mymagicrownumbers}
\newcommand\myrownumber{\stepcounter{mymagicrownumbers}\arabic{mymagicrownumbers}.}

\newcommand{\Prob}{\ensuremath{\operatorname{Pr}}}

\usepackage{caption}
\usepackage[caption=false,font=footnotesize]{subfig}
\usepackage{url}

\begin{document}
%
\title{Three-terminal Erasure Source-Broadcast with Feedback}
%
%


\author{\IEEEauthorblockN{Louis Tan, Kaveh Mahdaviani and Ashish Khisti~\IEEEmembership{Member,~IEEE}}
\thanks{L.~Tan, K.~Mahdaviani and A.~Khisti are with the Dept.\ of Electrical and Computer Engineering, University of Toronto, Toronto, ON, Canada (e-mail: louis.tan@mail.utoronto.ca, mahdaviani@cs.toronto.edu, akhisti@ece.utoronto.ca). 
Part of this work was presented at the 2015 International Symposium on Information Theory in Hong Kong~\cite{TMKS_ISIT15}.}%
}%
\maketitle

\begin{abstract}
%

We study the effects of introducing a feedback channel in the erasure source-broadcast problem for the case of three receivers.  In our problem formulation, we wish to transmit a binary source to three users over the erasure broadcast channel when causal feedback is universally available to the transmitter and all three receivers.  Each receiver requires a certain fraction of the source sequence, and we are interested in the minimum latency, or transmission time, required to serve them all.  

We propose a queue-based hybrid digital-analog coding scheme that achieves optimal performance for the duration of analog transmissions.  
We characterize the number of analog transmissions that can be sent as the solution to a linear program and furthermore give sufficient conditions for which optimal performance can be achieved by all users.  
In some cases, we find that users can be point-to-point optimal regardless of their distortion constraints.  

When the analog transmissions are  insufficient in meeting user demands, we propose two subsequent coding schemes.  The first uses a queue preprocessing strategy for channel coding.  The second method is a novel \emph{chaining algorithm}, which involves the transmitter targeting point-to-point optimal performance for two users as if they were the only users in the network. Meanwhile, the third user simultaneously builds ``chains'' of symbols that he may at times be able to decode based on the channel conditions, and the subsequent reception of additional symbols.  

\end{abstract}



%
\IEEEpeerreviewmaketitle

\section{Introduction\label{sec:intro_three_users}}

In this paper, we study the erasure source-broadcast problem with feedback for the case of three users in what is the sequel to~\cite{TMKS_TIT20}, which studied the analogous problem for the case of \emph{two} users.  In~\cite{TMKS_TIT20}, it was shown that the source-channel separation bound can always be achieved when a binary source is to be sent over the erasure broadcast channel with feedback to two users with individual erasure distortion constraints.  This is true whether a feedback channel is available for both receivers or when a feedback channel is available to only the stronger receiver.  The coding scheme relied heavily on the transmission of \emph{instantly-decodable, distortion-innovative} symbols~\cite{TMKS_TIT20,SorourValaee15,SorourValaee10}.  For the case of two receivers, each with a feedback channel, we found that transmitting such symbols was always possible.  

In this work, we again utilize uncoded transmissions that are instantly-decodable, and distortion-innovative.  The zero latency in decoding uncoded packets has benefits in areas in which packets are instantly useful at their destination such as applications in video streaming and disseminating commands to sensors and robots~\cite{SorourValaee15,SorourValaee10}.
However, for the case of three users, we find that the availability of sending such uncoded transmissions is not always possible.  Instead, it is random and based on the channel noise.  When the opportunity to send such packets is no longer possible, we design two alternative coding schemes.  The first is based on channel coding, while the second is a novel chaining algorithm that performs optimally for a wide range of channel conditions we have simulated.

Typically, in order to find opportunities to send instantly-decodable, distortion-innovative symbols, we set up queues that track which symbols are required by which group of users and subsequently find opportunities for network coding based on these queues.  In our analysis of our coding schemes, we propose two new techniques for analyzing queue-based opportunistic network coding algorithms. The first is in deriving a linear program to solve for the number of instantly-decodable, distortion-innovative symbols that can be sent. The second technique is in using a Markov rewards process with impulse rewards and absorbing states to analyze queue-based algorithms.

Previously, there have been two primary techniques in deriving a rate region for queue-based coding schemes, which we have found to be insufficient for our purposes.  The first technique was derived from a channel coding problem involving erasure channels with feedback and memory~\cite{HeindlmaierBidokhti16}.  In~\cite{HeindlmaierBidokhti16}, the authors use a queue stability criterion to set up a system of inequalities involving the achievable rates and flow variables that represent the number of packets that are transferred between queues.  When this technique is applied to the chaining algorithm of Section~\ref{subsubsec:chaining_algorithm} however, we find that our system of inequalities is over-specified.  That is, the large number of queues leaves us with more equations than unknown variables, and we are unable to solve for a rate region.  

The technique of using a Markov rewards process to analyze the rate region of a queue-based algorithm was also used in~\cite{GeorgiadisTassiulas_Netcod09}.  In their work, the transmitter and receivers were modelled with a Markov rewards process.  The authors solved for the \emph{steady-state} distribution and were able to derive a rate region by considering the number of rewards accumulated per timeslot and thus the number of timeslots it would take to send a group of packets.  In our analysis in Section~\ref{subsec:chaining_algorithm_analysis} however, we find that our transmitter and receiver must be modelled with an \emph{absorbing} Markov rewards process.  Thus, it is the \emph{transient} behaviour that is important and we cannot use a steady-state analysis.

In Section~\ref{subsec:chaining_algorithm_analysis}, we modify the analysis for that of an \emph{absorbing} Markov rewards process.  Specifically, we solve for the expected accumulated reward before absorption of a Markov rewards process with impulse rewards and absorbing states. The analysis does not involve the optimization of any parameters.  Rather, we use the analysis of the absorbing Markov rewards process to state sufficient conditions for optimality.  We illustrate the operational significance of this analysis in Section~\ref{sec:operational_meaning} where we show how Theorem~\ref{thm:chaining_sufficient} can be used to delineate regions of optimality.  Finally, in Section~\ref{subsec:instantly_decodable}, we also show how to formulate a linear program to derive a rate region by solving for the number of instantly-decodable, distortion innovative transmissions that can be sent.

\section{{System Model}}
\label{sec:system_model_three_users}






We study the problem of communicating a binary memoryless source $\{S(t)\}_{t=1,2, \ldots}$ to three users over an erasure broadcast channel with causal feedback that is universally available to both the transmitter and all receivers.  
%
%
The source produces equiprobable symbols in $\mathcal{S}=\{0,1\}$ and is communicated by an encoding function that produces the channel input sequence $X^{W} = (X(1),  \dots , X(W))$, where $X(t)$ denotes the $t^{\mathrm{th}}$ channel input taken from the alphabet $\mathcal{X} = \{0, 1\}$.  We assume that 
$X(t)$ is a function of the source as well as the channel outputs of all users prior to time $t$.  


Let $Y_{i}(t)$ be the channel output observed by user $i$ on the $t^{\mathrm{th}}$ channel use for $i \in \{1, 2, 3\}$. 
We let $Y_{i}(t)$ take on values in the alphabet $\mathcal{Y} = \{0, 1, \star\}$ so that an erasure event is represented by `$\star$'.   
For $W \in \mathbb{N}$,  let $[W]$ denote the set $\{1, 2, \ldots, W\}$. We associate user~$i$ with the state sequence $\{Z_i(t)\}_{t \in [W]}$, which represents the noise on user~$i$'s channel, where $Z_i(t) \in \mathcal{Z} \triangleq \{0,1\}$,  and  $Y_i(t)$ will be erased if $Z_{i}(t) = 1$ and $Y_{i}(t) = X(t)$ if $Z_i(t) = 0$.  The channel we consider is memoryless in the sense that $Z_{i}(t)$ is drawn i.i.d.\ from a $\operatorname{Bern}(\epsilon_{i})$ distribution, where
$\epsilon_i$ denotes the erasure rate of the channel corresponding to user~$i$.

The problem we consider involves causal feedback  that is universally available.  That is, at time $T$, we assume that $\{Z_1(t), Z_2(t), Z_3(t)\}_{t=1, 2, \ldots, T-1}$ is available to the transmitter and all receivers. 
After $W$ channel uses, user $i$ utilizes the feedback and his own channel output to reconstruct the source as a length-$N$ sequence, denoted as $\hat{S}_{i}^{N}$.  We will be interested in a fractional recovery requirement so that each symbol in $\hat{S}_{i}^{N}$ either faithfully recovers the corresponding symbol in $S^{N}$, or otherwise a failure is indicated with an erasure symbol, i.e., we do not allow for any bit flips.

More precisely, we choose the reconstruction alphabet $\mathcal{\hat{S}}$ to be an augmented version of the source alphabet so that $\mathcal{\hat{S}} = \{0, 1, \star\}$, where the additional `$\star$' symbol indicates an erasure symbol.  Let $\mathcal{D} = [0,1]$ and $d_i \in \mathcal{D}$ be the distortion user $i$ requires.  We then express the constraint that an achievable code ensures that each user $i \in \{1, 2, 3\}$ achieves a fractional recovery of $1 - d_{i}$ with the following definition.

\begin{mydef}
\label{def:code_two_users_feedback}
	An $(N, W, d_{1}, d_{2}, d_{3})$ code for source $S$ on the erasure broadcast channel with \emph{universal} feedback consists of	
	\begin{enumerate}
		\item a sequence of encoding functions $f_{t, N} : \mathcal{S}^{N} \times  \prod_{j = 1}^{3} \mathcal{Z}^{t-1} \to \mathcal{X}$ for $t \in [W]$, such that $X(t) = f_{t, N}(S^{N}, Z_1^{t -1}, Z_2^{t -1}, Z_3^{t -1})$, and
		
		\item three decoding functions $g_{i,N} : \mathcal{Y}^{W} \times \mathcal{Z}^{3W} \to \mathcal{\hat{S}}^{N}$ such that for $i \in \{1, 2, 3\}$, $\hat{S}_{i}^{N} = g_{i,N}(Y_{i}^{W}, Z_1^{W}, Z_2^{W}, Z_3^{W})$, and
		\begin{enumerate}
			\item $\hat{S}_{i}^{N}$ is such that for $t \in [N]$, if $\hat{S}_{i}(t) \neq S(t)$, then $\hat{S}_{i}(t) = \star$,
			\item $\mathbb{E}   \left\vert{\{t \in [N] \mid \hat{S}_{i}(t) = \star\}}\right\vert \leq N d_{i}$.
		\end{enumerate}
		 	
	\end{enumerate}
	
\end{mydef}

We again mention that in our problem formulation, we assume that all receivers have causal knowledge of $(Z_{1}^{t-1}, Z_{2}^{t-1}, Z_{3}^{t-1})$ at time $t$.  That is, each receiver has causal knowledge of which packets were received at each destination.  This can be made possible, for example, through the control plane of a network.

We define the {\it latency} that a given code requires before all users can recover their desired fraction of the source as follows.
\begin{mydef}
\label{def:latency_two_users}
	The latency, $w$, of an~$(N, W, d_{1}, d_{2}, d_{3})$ code is the number of channel uses per source symbol that the code requires to meet all distortion demands, i.e., $w = W/N$.
\end{mydef}
Our goal is to characterize the achievable latencies under a prescribed distortion vector,
as per the following definition.
\begin{mydef}
\label{def:achievable_general_feedback}
	Latency $w$ is said to be $(d_{1}, d_{2}, d_{3})$-achievable over the erasure broadcast channel if for every $\delta > 0$, there exists for sufficiently large $N$, an $(N, wN, \hat{d}_{1}, \hat{d}_{2}, \hat{d}_{3})$ code such that for all $i \in \{1, 2, 3\}$, $d_{i}+\delta \geq \hat{d}_{i}$.
	
	

\end{mydef}

\subsection{Organization of Paper}

The remainder of this paper is organized as follows.  In Section~\ref{sec:three_users}, we present an achievable coding scheme for the problem we have just defined.  In Section~\ref{sec:analysis}, we analyze the coding scheme and give sufficient conditions for optimality.  
Finally, we illustrate the performance of the coding schemes via simulations in Section~\ref{sec:simulations}. 

\section{Coding for Three Users}
\label{sec:three_users}

Our proposed coding scheme consists of two distinct parts and is described within Sections~\ref{subsec:instantly_decodable} and \ref{subsec:non_instant_coding} respectively. The first part involves only transmissions that are instantly decodable and distortion-innovative for all users, while the second part is executed if we are no longer able to send additional packets of this sort.  In~\cite{TMKS_TIT20}, the benefits of transmitting instantly decodable, distortion-innovative packets was explored.  In summary, a channel symbol is said to be instantly-decodable, distortion-innovative if any successfully received channel symbol can be immediately used to reconstruct a single source symbol that was hitherto unknown~\cite{TMKS_TIT20}.  In addition to minimizing decoding delays, incorporating such symbols is advantageous since if each user $i$ receives only such symbols for $wN$ transmissions, then they will all simultaneously be able to achieve their optimal distortion $D^{*}(\epsilon_i, w) \triangleq 1 - w(1 - \epsilon_i)$, which is obtained from the source-channel separation theorem.  This is because precisely $wN(1 - \epsilon_i)$ channel symbols are received after $wN$ transmissions have been made, and so if each transmission involved an instantly-decodable, distortion-innovative symbol, then $D^{*}(\epsilon_i, w)$ can trivially be achieved.

When we are no longer able to continue sending instantly-decodable, distortion-innovative symbols however, we describe two alternative approaches for the second part of our coding sheme.  In Section~\ref{subsec:channel_coding}, we propose a channel coding scheme that optimizes over the common and private messages we can channel-code to the group of users in order for their distortion constraints to be met.  Alternatively, in Section~\ref{subsubsec:chaining_algorithm}, we describe a \emph{chaining algorithm} that can be used instead whereby the transmitter targets point-to-point optimal performance for two users as if they were the only users in the network. Meanwhile, the third user simultaneously builds ``chains'' of symbols that he may at times be able to decode based on the channel conditions, and the subsequent reception of additional symbols.  


\subsection{Instantly Decodable, Distortion-Innovative Transmissions}
\label{subsec:instantly_decodable}

The first part of the code involves
an initial phase that transmits each source symbol until at least one user receives it.  No further processing of source symbol $S(t)$ is done if it was received by all users.  Let $\mathcal{E}(T) \subset \mathcal{U} \triangleq \{1, 2, 3\}$ represent the set of users whose channel output was an erasure at time $T$.  We will at times drop the time index $T$ when it is obvious from the context.  Then after the transmission at time $T$, we place $S(t)$ in queue $Q_{\mathcal{E}(T)}$ so that, in general, $Q_{U}$ maintains a record of which symbols were not received by all users $i \in U$.

The algorithm contains a subroutine that acts as an \emph{event handler}.  In the event that any user~$i$'s distortion constraint is met during or after the uncoded transmissions, the algorithm's control flow is passed to this subroutine.  Within this subroutine, we simply use the two-user algorithm outlined in~\cite{TMKS_TIT20} to serve the remaining users.  In~\cite{TMKS_TIT20}, it was shown that for the two-receiver version of the problem we consider, the optimal latency can be achieved regardless of either user's distortion constraint.  The coding scheme maintained queues of symbols $Q_{\{1\}}$, $Q_{\{2\}}$, and $Q_{\{1, 2\}}$ where the symbols in $Q_U$ are instantly-decodable, distortion-innovative for all users $i \in U$.  To incorporate~\cite{TMKS_TIT20} into our current coding scheme in the event that user~$i$ has met their distortion constraint, we first discard $Q_i$ since users $j$ and $k$ have already received all symbols from this queue for $j, k \in \mathcal{U} \setminus\{i\}, j\neq k$.  We then merge queues $Q_j$ and $Q_{\{i, j\}}$ and queues $Q_k$ and $Q_{\{i, k\}}$, since user~$i$ is no longer relevant, and symbols in $Q_{\{i, j\}}$ can be regarded as symbols that only user~$j$ needs at this point.  Finally, since $Q_{\{j,k\}}$ contains source symbols that neither user has received so far, we can  treat them as source symbols that have yet to be sent uncoded in the algorithm of~\cite{TMKS_TIT20}.

If there are no users with distortion constraints met however, the algorithm continues by recognizing network coding opportunities and sending linear combinations in a manner that is similar to that in~\cite{TMKS_TIT20}.  First, consider $Q_1$ and $Q_{\{2, 3\}}$.  If they are both non-empty, let $q_1 \in Q_1$ and $q_{2,3} \in \Q{2,3}$.  Notice that if we transmit $q_1 \oplus q_{2,3}$, since user~1 has access to $q_{2,3}$, and user~2 and user~3 have access to $q_1$, every user is able to receive an instantly decodable, distortion-innovative symbol.  A symbol is removed from a queue if any user for whom the queue is intended receives the linear combination.  However, if we have a type of situation where only user~2 receives the linear combination but not user~3, then $q_{2, 3}$ is transferred from \Q{2,3} to $Q_3$, since only user~3 is in need of this symbol now.  We continue this procedure with queues $Q_2$ and \Q{1,3} and queues $Q_3$ and \Q{1,2} to the extent that it is possible, i.e., for as long as the relevant queues required remain non-empty. 
Upon the emptying of these queues, we determine if we can send linear combinations of the form $q_1 \oplus q_2 \oplus q_3$,  which are similarly instantly decodable and distortion-innovative.  If this is possible, we again continue to do so until no longer possible.

We mention that in general, for $n$ users, we can similarly maintain queues that manage which symbols are erased by which users.  In this case, we can send an instantly decodable, distortion-innovative symbol if there exists non-empty queues $Q_{\Gamma_1}, Q_{\Gamma_2}, \ldots, Q_{\Gamma_m}$ such that $\cup_{i=1}^{m}\Gamma_i = [n]$ and $\sum_{i=1}^{m}|\Gamma_i| = n$, where $[n]$ denotes the set $\{1, 2, \ldots, n\}$.  These two conditions ensure that the index~$j$ for user~$j$ appears in exactly one $\Gamma_l, l \in [m]$.  Thus, user~$j$ is in possession of all symbols in these queues except for $Q_{\Gamma_l}$.
We then send the linear combination $\sum_{i = 1}^{m} q_{\Gamma_i}$, where $q_{\Gamma_i} \in Q_{\Gamma_i}$. We note however, that this requires an amount of queues that is exponential in the number of users.

\subsection{Non-Instantly-Decodable, Distortion-Innovative Coding}
\label{subsec:non_instant_coding}

When we have exhausted the queues that allow us to transmit instantly decodable, distortion-innovative symbols, we are left with two possibilities for the types of queues remaining.  By assumption, we have that for all $i \in \mathcal{U}$, either $Q_i$ or $Q_{\mathcal{U} \setminus \{i\}}$ is empty, and there exists an $l \in \mathcal{U}$ such that $Q_l$ is empty.  Thus, after the stopping condition of the algorithm in Section~\ref{subsec:instantly_decodable} has been reached, we are either left with queues \Q{1,2}, \Q{1,3} and \Q{2,3}, or queues $Q_i$, \Q{i,j} and \Q{i,k} (we will often use the indices $i, j$ and $k$ to refer to unique elements in $\mathcal{U}$).   We assume that the latter is the case and propose two methods to address this.  Analogous methods (omitted for brevity) can be used for the former case.

The first method involves preprocessing the remaining queues before using a channel coding scheme to satisfy the users' remaining demands.    The second method is a ``chaining algorithm,'' which involves the transmitter using the algorithm in~\cite{TMKS_TIT20} to send instantly-decodable, distortion-innovative transmissions to two users as if they were the only users in the network.  Meanwhile, the third user simultaneously builds ``chains'' of symbols that he may at times be able to decode based on the channel conditions, or with the subsequent reception of additional symbols.  We begin by describing the queue preprocessing method for channel coding.

\subsubsection{Queue Preprocessing for Channel Coding}
\label{subsec:channel_coding}

The channel coding scheme we employ is based on~\cite{GGT}, which outlines a method of \emph{losslessly} communicating a set of messages to $n$ users when feedback is available at the transmitter.  As input, it takes a set of queues $\{Q_U \mid U \subseteq \mathcal{U}\}$ where $U$ is a subset of users each of whom must losslessly reconstruct \emph{all} symbols in $Q_U$.

We preprocess queues $Q_i$, \Q{i, j} and \Q{i,k} to determine queues \tQo{i}, \tQo{j}, \tQo{k}, \tQ{i,j}, \tQ{i,k} that will be passed as input to the channel coding algorithm.  Let $\myd{R}{T}$ denote the number of symbols taken from queue $Q_R$ and placed in \tQo{T}.  Thus, we have, for example, that $| \tQo{i}| = \myd{i}{i} + \myd{i,j}{i} + \myd{i,k}{i}$.  Given that user $i$ has received $r_i$ symbols so far, we determine all $\delta$'s by solving the linear program in~\eqref{eq:delta_min}.
\begin{figure}
\begin{equation}
\label{eq:delta_min}
\begin{aligned}
	& \underset{\pmb{\delta}}{\text{min}}
	&& \frac{\myd{i}{i} +  \myd{i,j}{i} + \myd{i,k}{i}}{1 - \epsilon_i} + \frac{\myd{i,k}{k}}{1 - \epsilon_{k}} +  \frac{\myd{i,j}{j}}{1 - \epsilon_{j}} \\
	&&& \qquad + \frac{\myd{i,k}{i,k}}{1 - \max(\epsilon_i, \epsilon_k)} + \frac{\myd{i,j}{i,j}}{1 - \max(\epsilon_i, \epsilon_j)}
\end{aligned}
\end{equation}
\[
\begin{aligned}
	& \text{subject to}
	&& \pmb{\delta} \succeq 0 \\
	&&&  \myd{i}{i}  \leq | Q_i |, \\
	&&&  \myd{i,k}{i} +  \myd{i,k}{k} + \myd{i,k}{i,k} \leq | \Q{i, k} |, \\
	&&&  \myd{i,j}{i} +  \myd{i,j}{j} + \myd{i,j}{i,j} \leq | \Q{i, j} |, \\
	&&& \myd{i}{i} + \myd{i,j}{i} + \myd{i,k}{i} + \myd{i, j}{i, j} + \myd{i,k}{i,k} \geq 1 - d_i - r_i, \\
	&&&  \myd{i,j}{j} + \myd{i,j}{i,j}  \geq 1 - d_j - r_j, \\
	&&&  \myd{i,k}{k} + \myd{i,k}{i,k}  \geq 1 - d_k - r_k,
\end{aligned}
\]
\end{figure}
%
Here, each term in the objective function represents the minimum latency required to process its corresponding queue, i.e., the first term represents the latency required to process $Q_i$, the second is for $Q_k$, etc.  Furthermore, we have that user $i$ is able to decode  \myd{R}{T} symbols if $i \in T$, and so given that this user has already received~$r_i$ symbols, the distortion constraints in~\eqref{eq:delta_min} follow.  

After solving~\eqref{eq:delta_min}, we run the channel coding algorithm on the newly determined queues, after which, we have that the total transmission time is the sum of the time required for this algorithm as well as the one from Section~\ref{subsec:instantly_decodable}.  Finally, we mention that a generalization of this section's approach for $n$ users is an ongoing work.

\subsubsection{A Chaining Algorithm}
\label{subsubsec:chaining_algorithm}

In this section, we propose an alternative to using queue preprocessing for channel coding as described in the previous section.  Recall that we invoke this algorithm when we have exhausted all opportunities to send instantly-decodable, distortion-innovative symbols which happens when we are left with queues $\Q{1,2}, \Q{1,3}$ and $\Q{2,3}$, or queues $Q_i$, \Q{i,j} and \Q{i,k}.  We again assume that the latter is the case.

We first consider if we were to ignore user~$i$ and simply use the coding scheme of~\cite{TMKS_TIT20} to target point-to-point optimal performance for the remaining users.  Specifically, since queues $\Q{i,j}$ and $\Q{i,k}$ remain, we can send linear combinations of the form $q_{i,j} \oplus q_{i, k}$, where $q_{i, j} \in \Q{i, j}$ and $q_{i, k} \in \Q{i, k}$.  Recall that user~$j$ has received all symbols in $\Q{i,k}$, and user~$k$ has similarly received all symbols in $\Q{i,j}$.  In each transmission, each user can therefore subtract off one of the symbols from the linear combination so that only one symbol, the symbol of their interest, is remaining.  The transmitter can use the feedback of both users to replace $q_{i, j}$ or $q_{i, k}$ whenever it is received by user~$j$ or $k$ respectively, thus making this strategy simultaneously optimal for both user~$j$ and user~$k$ as mentioned in~\cite{TMKS_TIT20}.

Now, consider what user~$i$ can do in the midst of being ignored.  Specifically, consider if the following transmissions and erasure patterns were to occur as in Table~\ref{tab:hypothetical_chaining}.  In the first row of the table, we assume that at time $t = 1$, the linear combination $q_{i, j} \oplus q_{i, k}$ is sent.  We also assume that at this time, user~$i$ and user~$j$ received this transmission while user~$k$ did not, which is indicated from the fact that the channel noises take on values $(Z_{i}(t), Z_{j}(t), Z_{k}(t)) = (0, 0, 1)$ in the first row (see Section~\ref{sec:system_model_three_users}).  Since user~$j$ received the transmission, we replace $q_{i, j}$ with $\hat{q}_{i, j} \in \Q{i, j}$ and send the new linear combination $\hat{q}_{i, j} \oplus q_{i, k}$ at time $t = 2$.  At this time, only user~$i$, the user being ignored, receives the transmission.  Thus, since we are ignoring user~$i$, there is no need to replace any of the symbols, and instead we should retransmit the linear combination at time $t = 3$.  Notice however, if instead of retransmitting $\hat{q}_{i, j} \oplus q_{i, k}$ at time $t = 3$, we instead send $2\hat{q}_{i, j} \oplus q_{i, k}$, which is a \emph{different} linear combination of the same two symbols sent at $t = 2$.  This comes at no cost to user~$j$ nor user~$k$, since they would be able to decode regardless of which linear combination is sent. Let us further assume that at time $t = 3$, we again have that only user~$i$ has received the transmission.  Then since user~$i$ has received all transmissions so far, he will have received three independent equations in three unknown variables.  Thus he is able to decode all symbols sent so far despite the fact that the transmitter is targeting point-to-point optimal performance for the \emph{other} two users.  

\begin{table}
	\begin{center}
		\begin{tabular}{l c c c c}
			$t$              & $X(t)$ & $Z_{i}(t)$& $Z_{j}(t)$ & $Z_{k}(t))$  \\
			\hline
			1 & $q_{i, j} \oplus q_{i, k}$ & 0 & 0 & 1 \\
			2 & $\hat{q}_{i, j} \oplus q_{i, k}$ & 0 & 1 & 1 \\
			3 & $2\hat{q}_{i, j} \oplus q_{i, k}$ & 0 & 1 & 1 \\
		\end{tabular}
	\end{center}
	\caption{A hypothetical sequence of transmissions and channel noises when targeting point-to-point optimal performance for user~$j$ and user~$k$.}	
	\label{tab:hypothetical_chaining}
\end{table}


We see that the chaining algorithm ensures that each linear combination received by user~$i$ has at least one symbol in common with at least one other linear combination that was previously received.  For example, in Table~\ref{tab:hypothetical_chaining}, the first two transmissions have $q_{i, k}$ in common and the last two transmissions have two symbols in common.  By maintaining this property, we ensure that at any time, user~$i$ has to decode no more than one symbol in order to be able to back-substitute and solve for the remaining symbols received so far.  

We visualize this with the analogy of a chain of links.  Each linear combination received is a link in a chain, and links are coupled by the fact that they have at least one symbol in common.  If one source symbol in the chain is revealed to user~$i$, it can be used to solve for all symbols in one of the links.  In turn, that link has a symbol in common with a second link, and so all symbols in the second link can also be solved for.  This process continues until the entire chain of symbols can be decoded.  




We model the chaining algorithm as a Markov rewards process with absorbing states.  The process starts in an initial state, after which, transitions between states occur at every timeslot.  The determination of which transition was actually taken depends only on the channel noise.  The process runs until one of two absorbing states is reached, after which, we are either able to decode an entire chain of symbols, or we have finished creating a chain of symbols that can be decoded with the reception of one more linear combination of symbols.  In the example of Table~\ref{tab:hypothetical_chaining}, we would reach the absorbing state in which we are not able to decode the chain of symbols if at $t = 3$, we instead had that only users~$j$ and~$k$ received $X(3)$.  In this case, we would have to replace both symbols being sent at $t = 4$ in order to be point-to-point optimal for users~$j$ and~$k$.  Thus, user~$i$ would not have enough equations to decode, and we would instead reset the Markov process such that user~$i$ would start building another chain when new symbols are sent.

\begin{table}
	\begin{center}
		\begin{tabular}{l | l | p{4.5cm} c }
			$State$              & $X(t)$ -- Transmission at time $t$ & Precondition at $t$  \\
			\hline
			1 & $q_{i, j} \oplus q_{i, k}$ (Initial state) & $|\chainsetarg{t-1} \cap \xsett| = 0$\\
			&& State~1 occupied at time $t-1$\\ 
			2 & $\hat{q}_{i, j} \oplus q_{i, k}$ & $|\chainsetarg{t-1} \cap \xsett| = 1$ \\ 
			3 & $q_{i, j} \oplus \hat{q}_{i, k}$ & $|\chainsetarg{t-1} \cap \xsett| = 1$ \\ 
			4 & $2q_{i, j} \oplus q_{i, k} $  & $|\chainsetarg{t-1} \cap \xsett| = 2$ \\ 	
			5 & None (Non-decoding absorbing state)  & $|\chainsetarg{t-1} \cap \xsett| = 0$\\ 
			6 & None (Decoding absorbing state)  & State~4 occupied at time $t-1$\\ 
		\end{tabular}
	\end{center}
	\caption{A description of states in the Markov rewards process.  Assume that $q_{i, j} \oplus q_{i, k}$ was sent at time $t-1$.  If the state indicated in the first column is entered at time $t$, then the transmission indicated by the second column is sent at time $t$.  We also show the precondition necessary to enter the state in the third column.}	
	\label{tab:table_of_states}
\end{table}

In addition to the two absorbing states, the Markov rewards process involves four additional \emph{transient} states.  We enumerate and describe all states in Table~\ref{tab:table_of_states}.  In this table, we assume that $q_{i, j} \oplus q_{i, k}$ was sent at time $t-1$.  Then if state $l \in \Omega \triangleq \{1, 2, \ldots, 6\}$ is entered at time $t$, we show $l$ in the first column and $X(t)$, the transmission at time $t$, in the second column.  We also show a necessary precondition for entering state $l$ at time $t$ in the third column.

Before we explain the preconditions, let us first define some notation.  Let $\chainsett$ represent the set of symbols appearing in the linear combinations received by user~$i$ after having listened to all transmissions from the beginning of the chaining algorithm up to and including time $t$.  In the example of Table~\ref{tab:hypothetical_chaining}, after time $t = 1, 2, 3$, we have $\chainsetarg{1} = \{q_{i, j}, q_{i, k}\}$, $\chainsetarg{2} = \{q_{i, j}, q_{i, k}, \hat{q}_{i, j}\}$ and $\chainsetarg{3} = \{q_{i, j}, q_{i, k}, \hat{q}_{i, j}\}$.  Let $\xsett$ represent the set of symbols appearing in the linear combination in $X(t)$.  In the example of Table~\ref{tab:hypothetical_chaining}, we have that $\xsetarg{1} = \{q_{i, j}, q_{i, k}\}$, $\xsetarg{2} = \{\hat{q}_{i, j}, q_{i, k}\}$, and $\xsetarg{3} = \{\hat{q}_{i, j}, q_{i, k}\}$. 
With this notation, we now describe the preconditions of Table~\ref{tab:table_of_states}.


The first state is the initial state that sends the initial linear combination of symbols.  After the first outbound transition from the initial state, the initial state is not returned to at any time during the evolution of the Markov reward process unless an absorbing state is reached and we restart the chain-building process.  Since user~$i$ has not yet received any equations and we begin by transmitting a linear combination of two symbols, we have the precondition that $\chainsetarg{t-1} = \varnothing$, $\xsett = \{q_{i, j}, q_{i, k}\}$ and $|\chainsetarg{t-1} \cap \xsett| = 0$.

The second state is entered only if the symbol destined for user~$j$ needs to be replaced in the next linear combination to be sent and the chain-building process has begun.  In order to ensure that we can continue building the chain of symbols, we require that at least one symbol being sent at time $t$ appears in a linear combination already received by user~$i$, i.e.,  $|\chainsetarg{t-1} \cap \xsett| = 1$.  Similarly the third state is entered only if the symbol destined for user~$k$ needs to be replaced in the next linear combination to be sent and the chain-building process has already started.  The precondition for this state is analogous to its counterpart in State~2.  

The fourth state is entered if \emph{only} user~$i$ received the previous transmission, in which case we need to send a new linear combination of the \emph{same} two symbols in the previous timeslot.  In this case, we see that since the transmission is just a different linear combination of the previous two symbols sent, $|\chainsetarg{t-1} \cap \xsett| = 2$.

Finally, the two remaining states are the absorbing states.  State~5 is the absorbing state we occupy if the constraint of being point-to-point optimal for users~$j$ and~$k$ prevent user~$i$ from continuing the chain-building process.  That is, we are not able transmit the next linear combination such that one of the symbols appearing in the combination also appears in a previously received equation, i.e., $|\chainsetarg{t-1} \cap \xsett| = 0$.  Lastly, State~6 is the absorbing state we occupy if we reach the point when we are able to decode all symbols in the chain.  It is reached if the transmission sent in State~4 was received by user~$i$.

Having described the purpose of each state, we now describe the transitions between states in detail beginning with State~1.  Let $\zall = (Z_{i}(t), Z_{j}(t), Z_{k}(t))$.  For each state, we use a table to describe the outbound transitions based on values of $\zall$.  In addition, we show the reward accumulated for each transition.  For example, in the column with the heading $\rhosub{j}$ in Table~\ref{tab:state1_transitions_all}, we show the reward accumulated by user~$j$ at each possible outbound transition from State~1 based on the channel noise.  We note that although $\rhosub{j}$ depends on both $\zall$ and the inbound and outbound states, for notational convenience, we omit explicitly stating this dependence.  The reward can be accumulated by a user, or a queue.  The notation for rewards is summarized in Table~\ref{tab:reward_variables}.  In our description of each state, we also define the queue $\Qstar$, which is the queue containing prioritized symbols that only user~$i$ requires.  These symbols are prioritized because the reception of this symbol can decode an entire chain of symbols.  

In the following descriptions, the reader can confirm that the chaining algorithm is point-to-point optimal for users~$j$ and~$k$.  Each time a symbol is received by one of the users, it is replaced by another instantly-decodable, distortion-innovative symbol.



\begin{table}
	\begin{center}
		\begin{tabular}{l |  p{8cm}}
			Variable              & Description  \\
			\hline
			$\rhosub{u}(l, m)$ & Reward representing the number of symbols that can be decoded by user~$u$ after a transition from state $l$ to $m$ for $u \in \{j , k\}$ \\
			$\rhosub{Q}(l, m)$ &Reward representing the number of symbols placed in queue $Q$ after a transition from state $l$ to $m$.  \\
			$\rhocarg{l}{m}$ &Reward representing the number of equations user~$i$ received after a transition from state $l$ to $m$. \\
		\end{tabular}
	\end{center}
	\caption{A legend for the reward variables.}	
	\label{tab:reward_variables}
\end{table}


\setcounter{mymagicrownumbers}{0} 

\begin{table}
	\begin{center}
		\begin{tabular}{c c l *{8}{c}}
			\multicolumn{10}{c}{{\bf State~1 Outgoing Transitions}} \\		
			& $\zall$              & $\Prob(\zall)$ & Next State & $\rhoj$ & $\rhok$ & $\rhoc$ & $\rhoqi$ & $\rhoqj$ & $\rhoqk$ & $\rhosub{\Qstar}$\\
			\hline
			\myrownumber & $(0, 0, 0)$ & $(1 -\epsilon_i)(1 -\epsilon_j)(1 - \epsilon_k)$ & 5 &1 & 1 & 1 & 0 & 0 & 0 & 1 \\
			\myrownumber & $(0, 0, 1)$ & $(1 -\epsilon_i)(1 -\epsilon_j)\epsilon_k$ & 2 &1 & 0 & 1 & 0 & 0 & 0 & 0 \\
			\myrownumber & $(0, 1, 0)$ & $(1 -\epsilon_i)\epsilon_j(1 - \epsilon_k)$ & 3 & 0 & 1 & 1 & 0 & 0 & 0 & 0 \\
			\myrownumber & $(0, 1, 1)$ & $(1 -\epsilon_i)\epsilon_j\epsilon_k$ & 4 & 0 & 0 & 1 & 0 & 0 & 0 & 0 \\
			\myrownumber & $(1, 0, 0)$ & $\epsilon_i(1 -\epsilon_j)(1 - \epsilon_k)$ & 5 & 1 & 1 & 0 & 2 & 0 & 0 & 0 \\
			\myrownumber & $(1, 0, 1)$ & $\epsilon_i(1 -\epsilon_j)\epsilon_k$ & 1 &1 & 0 & 0 & 1 & 0 & 0 & 0 \\											\myrownumber & $(1, 1, 0)$ & $\epsilon_i\epsilon_j(1 - \epsilon_k)$ & 1 &0 & 1 & 0 & 1 & 0 & 0 & 0 \\
			\myrownumber & $(1, 1, 1)$ & $\epsilon_i\epsilon_j\epsilon_k$ & 1 &0 & 0 & 0 & 0 & 0 & 0 & 0 \\
		\end{tabular}
	\end{center}
	\caption{A detailed table of all outgoing transitions from State~1 based on the channel noise realization.  The table includes the probability of each transition and the reward accumulated for each transition.}	
	\label{tab:state1_transitions_all}
\end{table}

\setcounter{mymagicrownumbers}{0} 
\begin{LaTeXdescription}
	\item[State~1] We enumerate all outgoing transitions from State~1, the initial state, based on all possible channel conditions in Table~\ref{tab:state1_transitions_all}.  
	\begin{enumerate}
		\item $\zall = (0, 0, 0)$.  In this case, all users have received the transmission.  We therefore see that $\rhoj = \rhok = 1$, since users~$j$ and $k$ can each decode a symbol.  Since this is the initial state however, user~$i$ has not received any previous equations, and therefore has only one equation in two unknown variables.  Since we must replace both symbols in the next transmission, we cannot continue building the current chain.  The next state is therefore State~5, the absorbing state in which we are not able to decode any symbols in the chain.  However, we do increase the number of equations received by user~$i$ by setting $\rhoc = 1$.  Furthermore, we arbitrarily place one of the symbols from the linear combination into $\Qstar$, the queue containing symbols that can decode an entire chain of symbols.  We set $\rhoqstar = 1$ to reflect this.  
		\item $\zall = (0, 0, 1)$. In this case, only users $i$ and~$j$ have received the transmission.  We set $\rhoj =  \rhoc =1$ to indicate that user~$j$ can decode one symbol, and user~$i$ has received one more equation.  Since the symbol for user~$j$ will be replaced in the next transmission, and the chain-building process has just begun, we therefore have that the next state is State~2.  
		\item $\zall = (0, 1, 0)$. Similar to the previous case, however, we now have that only users $i$ and~$k$ have received the transmission.  We set $\rhok =  \rhoc =1$ to indicate that user~$k$ can decode one symbol, and user~$i$ has received one more equation.  Since the symbol for user~$k$ will be replaced in the next transmission, and the chain-building process has just begun, we therefore have that the next state is State~3.  
		\item $\zall = (0, 1, 1)$.  In this case, only user~$i$ has received the transmission so we set only $\rhoc =1$.  For the subsequent timeslot, we must send a different linear combination of the same two symbols previously sent and so the next state is State~4.
		\item $\zall = (1, 0, 0)$.  In this case, both users~$j$ and~$k$ can decode a symbol so we set $\rhoj = \rhok =1$.  If $q_{i, j} \oplus q_{i, k}$ was sent in the previous timeslot, we now have that both $q_{i, j}$ and $q_{i, k}$ are required by only user~$i$.  We therefore place both these symbols in $Q_{i}$, and we set $\rhoqi = 2$ to reflect this.  Since both $q_{i, j}$ and $q_{i, k}$ will not be sent in the next transmission, we  have that the next state is the non-decoding absorbing state, State~5.
		\item $\zall = (1, 0, 1)$.  In this case, only user~$j$ received the transmission so we set $\rhoj = 1$.  Notice now that $q_{i, j}$ is no longer needed by user~$j$.  Although we must replace this symbol in the next transmission, the chain-building process has not yet begun and user~$i$ has still not received any equations involving $q_{i, j}$.  We therefore place this symbol in $Q_{i}$, set $\rhoqi = 1$ to reflect this, and return to State~1 in the next timeslot.
		\item $\zall = (1, 1, 0)$	.  Similar to the previous case, only user~$k$ has received the transmission and so we set $\rhok = \rhoqi = 1$, and return to State~1 in the next timeslot.
		\item $\zall = (1, 1, 1)$	.  In this case, no users have received the transmission.  No rewards are assigned and we return to State~1 to retransmit the same linear combination.
	\end{enumerate}		
\end{LaTeXdescription}


\setcounter{mymagicrownumbers}{0} 

\begin{table}
	\begin{center}
		\begin{tabular}{c c l *{8}{c}}
			\multicolumn{10}{c}{{\bf State~2 Outgoing Transitions}} \\		
			& $\zall$              & $\Prob(\zall)$ & Next State & $\rhoj$ & $\rhok$ & $\rhoc$ & $\rhoqi$ & $\rhoqj$ & $\rhoqk$ & $\rhosub{\Qstar}$\\
			\hline
			\myrownumber & $(0, 0, 0)$ & $(1 -\epsilon_i)(1 -\epsilon_j)(1 - \epsilon_k)$ & 5 &1 & 1 & 1 & 0 & 0 & 0 & 1 \\
			\myrownumber & $(0, 0, 1)$ & $(1 -\epsilon_i)(1 -\epsilon_j)\epsilon_k$ & 2 &1 & 0 & 1 & 0 & 0 & 0 & 0 \\
			\myrownumber & $(0, 1, 0)$ & $(1 -\epsilon_i)\epsilon_j(1 - \epsilon_k)$ & 3 & 0 & 1 & 1 & 0 & 0 & 0 & 0 \\
			\myrownumber & $(0, 1, 1)$ & $(1 -\epsilon_i)\epsilon_j\epsilon_k$ & 4 & 0 & 0 & 1 & 0 & 0 & 0 & 0 \\
			\myrownumber & $(1, 0, 0)$ & $\epsilon_i(1 -\epsilon_j)(1 - \epsilon_k)$ & 5 & 1 & 1 & 0 & 1 & 0 & 0 & 1 \\
			\myrownumber & $(1, 0, 1)$ & $\epsilon_i(1 -\epsilon_j)\epsilon_k$ & 2 &1 & 0 & 0 & 1 & 0 & 0 & 0 \\											\myrownumber & $(1, 1, 0)$ & $\epsilon_i\epsilon_j(1 - \epsilon_k)$ & 5 &0 & 1 & 0 & 0 & 0 & 0 & 1 \\
			\myrownumber & $(1, 1, 1)$ & $\epsilon_i\epsilon_j\epsilon_k$ & 2 &0 & 0 & 0 & 0 & 0 & 0 & 0 \\
		\end{tabular}
	\end{center}
	\caption{A detailed table of all outgoing transitions from State~2 based on the channel noise realization.  The table includes the probability of each transition and the reward accumulated for each transition.}	
	\label{tab:state2_transitions_all}
\end{table}

\setcounter{mymagicrownumbers}{0} 

\begin{LaTeXdescription}
	\item [State~2]  We enumerate all outgoing transitions from State~2 based on all possible channel conditions in Table~\ref{tab:state2_transitions_all}.  We remind the reader that State~2 is characterized by user~$i$ having already started the chain-building process and the transmitter having just replaced the symbol intended for user~$j$ in the linear combination being sent.  Therefore, the symbol intended for user~$k$ is the common symbol between $\chainsetarg{t-1}$ and $\xsett$ in the precondition of Table~\ref{tab:table_of_states}.  
	We now go over the possible transitions based on whether this new linear combination was received by any of the receivers.
	\begin{enumerate}
		\item $\zall = (0, 0, 0)$.  This case is analogous to its counterpart in State~1.  The subsequent state is again State~5, the non-decoding absorbing state, since user~$i$ has has one less equation required to decode the number of unknown variables in the chain.  The transition to the subsequent state and the rewards allocated are identical for when $\zall = (0,0,0)$ in State~1.  
		\item $\zall = (0, 0, 1)$. This case is analogous to its counterpart in State~1.  
		\item $\zall = (0, 1, 0)$. This case is analogous to its counterpart in State~1.  
		\item $\zall = (0, 1, 1)$. This case is analogous to its counterpart in State~1.  
		\item $\zall = (1, 0, 0)$.  This case is similar to its counterpart in State~1 with one minor but important difference due to the fact that the chain-building process has already begun.  Recall that in this case, we have that both symbols in the previous linear combination sent will have to be replaced in the next timeslot.  However, instead of placing both these symbols in $Q_i$, we should place one of them in $\Qstar$, since it can be used to decode an entire chain of symbols.  We therefore set $\rhoqstar = \rhoqi = 1$.
		\item $\zall = (1, 0, 1)$.  This case is analogous to its counterpart in State~1, however since the chain-building process has begun, we transition to State~2 again in the next timeslot.
		\item $\zall = (1, 1, 0)$.  This case is in direct contrast from its counterpart in State~1.  Recall from the beginning of the description for State~2 that the symbol intended for user~$k$ is the symbol in common between $\chainsetarg{t-1}$ and $\xsett$.  Furthermore, given that $\zall = (1, 1, 0)$, user~$k$ has now received this symbol and so we must replace it in the next timeslot.  However, doing so would not allow us to continue building the current chain because there will be no symbols in common between  $\chainsetarg{t}$ and $\xsetarg{t+1}$.  Therefore, we transition to the non-decoding absorbing state in the next timeslot and  
place the symbol that was intended for user~$k$ in $\Qstar$.
		\item $\zall = (1, 1, 1)$.  This case is analogous to its counterpart in State~1, except we transition back to State~2 in the next timeslot.
	\end{enumerate}
\end{LaTeXdescription}

\begin{LaTeXdescription}
	\item [State~3]  For brevity, we omit the details for State~3 and simply state that it is analogous to the discussion we have just given for State~2 with the roles of users~$j$ and~$k$ reversed.
\end{LaTeXdescription}

\setcounter{mymagicrownumbers}{0} 

\begin{table}
	\begin{center}
		\begin{tabular}{c c l *{8}{c}}
			\multicolumn{10}{c}{{\bf State~4 Outgoing Transitions}} \\		
			& $\zall$              & $\Prob(\zall)$ & Next State & $\rhoj$ & $\rhok$ & $\rhoc$ & $\rhoqi$ & $\rhoqj$ & $\rhoqk$ & $\rhosub{\Qstar}$\\
			\hline
			\myrownumber & $(0, 0, 0)$ & $(1 -\epsilon_i)(1 -\epsilon_j)(1 - \epsilon_k)$ & 6 &1 & 1 & 1 & 0 & 0 & 0 & 0 \\
			\myrownumber & $(0, 0, 1)$ & $(1 -\epsilon_i)(1 -\epsilon_j)\epsilon_k$ & 6 &1 & 0 & 1 & 0 & 0 & 1 & 0 \\
			\myrownumber & $(0, 1, 0)$ & $(1 -\epsilon_i)\epsilon_j(1 - \epsilon_k)$ & 6 & 0 & 1 & 1 & 0 & 1 & 0 & 0 \\
			\myrownumber & $(0, 1, 1)$ & $(1 -\epsilon_i)\epsilon_j\epsilon_k$ & 6 & 0 & 0 & 1 & 0 & 1 & 1 & 0 \\
			\myrownumber & $(1, 0, 0)$ & $\epsilon_i(1 -\epsilon_j)(1 - \epsilon_k)$ & 5 & 1 & 1 & 0 & 0 & 0 & 0 & 1 \\
			\myrownumber & $(1, 0, 1)$ & $\epsilon_i(1 -\epsilon_j)\epsilon_k$ & 2 &1 & 0 & 0 & 0 & 0 & 0 & 0 \\							\myrownumber & $(1, 1, 0)$ & $\epsilon_i\epsilon_j(1 - \epsilon_k)$ & 3 &0 & 1 & 0 & 0 & 0 & 0 & 0 \\
			\myrownumber & $(1, 1, 1)$ & $\epsilon_i\epsilon_j\epsilon_k$ & 4 &0 & 0 & 0 & 0 & 0 & 0 & 0 \\
		\end{tabular}
	\end{center}
	\caption{A detailed table of all outgoing transitions from State~4 based on the channel noise realization.  The table includes the probability of each transition and the reward accumulated for each transition.}	
	\label{tab:state4_transitions_all}
\end{table}

\setcounter{mymagicrownumbers}{0} 

\begin{LaTeXdescription}
	\item [State~4]  We enumerate all outgoing transitions from State~4 based on all possible channel conditions in Table~\ref{tab:state4_transitions_all}.  We remind the reader that State~4 is characterized by user~$i$ having just been the only receiver to receive the previous transmission.  Thus, the chain-building process has begun and in this timeslot, the transmitter will send a different linear combination of the previous two symbols sent.
	Finally, we note that this state has outbound transitions in which a source symbol may be placed in $Q_{j}$ or $Q_k$ \emph{after} time $t$.  In the event that one or both of these queues are nonempty, and depending on whether $Q_i$ is nonempty,  the transmitter may also have the option of temporarily suspending the chaining algorithm to send linear combinations of the form 
	\begin{subequations}	
	\label{eq:linear_comb_available}
		\begin{align}
			\label{eq:linear_comb_available_a}
			&q_{i} \oplus q_{j} \oplus q_{k},\\
			\label{eq:linear_comb_available_b}
			&q_j \oplus q_{i, k}, \\
			\label{eq:linear_comb_available_c}
			&q_{k} \oplus q_{i, j}, 
		\end{align}

	\end{subequations}
	which are instantly-decodable, distortion-innovative transmissions for all users.
	We now go over the possible state transitions based on whether the linear combination sent at time $t$ was received by any of the receivers.
	\begin{enumerate}
		\item $\zall = (0, 0, 0)$.  In this case, user~$i$ has now received as many equations as there are unknown variables in the chain.  Thus the subsequent state is State~6, the decoding absorbing state, and there is no need to set $\rhoqstar = 1$, since all variables in the chain have been accounted for.  
		\item $\zall = (0, 0, 1)$. In this case, user~$i$ can again decode. Furthermore, we place the source symbol intended for user~$k$ into $Q_k$ and set $\rhoqk = 1$.
		\item $\zall = (0, 1, 0)$. Similar to the previous case, user~$i$ can also decode given this channel noise realization, and we place the source symbol intended for user~$j$ into $Q_j$ and set $\rhoqj = 1$.
		\item $\zall = (0, 1, 1)$. In this case, user~$i$ can again decode, and we place the source symbol intended for users~$j$ and~$k$ into $Q_j$ and~$Q_k$ respectively and set $\rhoqj = \rhoqk = 1$.
		\item $\zall = (1, 0, 0)$.  This case is similar to its counterpart in State~2, however, since both symbols appearing in $X(t)$ have been accounted for, i.e., both symbols appearing in $X(t)$ are in the set $\chainsetarg{t-1}$, there is no need to set $\rhoqi = 1$.
		\item $\zall = (1, 0, 1)$.  This case is analogous to its counterpart in State~2, however, since the symbol just received by user~$j$ has been accounted for, i.e., since it is one of the symbols in $\chainsetarg{t-1}$, there is no need to set $\rhoqi = 1$.
		\item $\zall = (1, 1, 0)$	.  This case is analogous to its counterpart in State~2, however, since the symbol just received by user~$k$ has been accounted for, i.e., since it is one of the symbols in $\chainsetarg{t-1}$, there is no need to set $\rhoqi = 1$.
		\item $\zall = (1, 1, 1)$	.  This case is analogous to its counterpart in State~1, except we transition back to State~4 in the next timeslot.
	\end{enumerate}
\end{LaTeXdescription}

Finally, we conclude our discussion of the chaining algorithm by mentioning that if at any time an absorbing state in the Markov rewards process is reached and the distortion constraints of users~$j$ and~$k$ have yet to be met, then we simply restart the Markov rewards process and begin building a new chain.  If at any time during this process, one of the users has their distortion constraint met, then that user can leave the network and we are left with two users who we again serve with the two-user algorithm of~\cite{TMKS_TIT20}.  If one of the remaining two users is user~$i$, the user building the chain, then we would prioritize the transmission of symbols in $\Qstar$, since the reception of one of these symbols would lead to the decoding of an entire chain of symbols. 

\section{Analysis}
\label{sec:analysis}

In this section, we analyze the coding scheme proposed in Section~\ref{subsec:instantly_decodable}, and the chaining algorithm of Section~\ref{subsubsec:chaining_algorithm}.
For the analysis of the coding scheme in Section~\ref{subsec:instantly_decodable}, we show in Section~\ref{subsec:analysis_instantly_decodable} that the solution of a linear program characterizes the number of instantly decodable, distortion-innovative symbols that can be sent by the algorithm in Section~\ref{subsec:instantly_decodable}.
We furthermore give sufficient conditions for which all users can be point-to-point optimal.  In some cases, we find that users can be point-to-point optimal regardless of their distortion constraints.  Following this, in Section~\ref{subsec:chaining_algorithm_analysis} we analyze the chaining algorithm and give a sufficient condition for point-to-point optimal performance.

\subsection{Analysis of Coding Scheme of Section~\ref{subsec:instantly_decodable}}
\label{subsec:analysis_instantly_decodable}

Our analysis begins by first considering the systematic phase of our coding, which transmits the $N$ source symbols so that each symbol is recovered by at least one user.  Let $T_0$ be a random variable representing the number of transmissions required for this after being normalized by $N$.  In other words, $T_0$ will be said to be the \emph{latency} required.  Then it is clear that $T_0$ has expected value $\bar{T_0}$ given by
\begin{equation}
\label{eq:T0}
	\bar{T}_{0} = \frac{1}{1 - \epsilon_1\epsilon_2\epsilon_3}.
\end{equation}

During the systematic transmissions, we direct erased symbols to their appropriate queues as outlined in the algorithm of Section~\ref{subsec:instantly_decodable}.  Then, when the $NT_0$ transmissions are completed, we start sending linear combinations of the form $q_i \oplus q_{j,k}$, where $q_i \in Q_i, q_{j,k} \in \Q{j,k}$.
and $i,j,k\in \mathcal{U}$  are unique.

Let $T_i$ be a random variable representing the normalized number of transmissions we can send of the form $q_i \oplus q_{j,k}$.  We are able to do this so long as $Q_i$ and \Q{j,k} are non-empty.  Thus, we must bound the cardinality of these queues.  For \Q{j,k}, we have that a symbol is added during the systematic transmissions whenever user~$i$ receives a symbol that is erased at users~$j$ and~$k$.  Symbols are no longer added to this queue after the systematic transmissions, and so we have that $\Q{j,k}^{+}$, the  expected maximum normalized cardinality of \Q{j,k}, is given by
\begin{equation}
\label{eq:Qjk}
	\Q{j,k}^{+} = \bar{T}_0(1 - \epsilon_i)\epsilon_j\epsilon_k.
\end{equation}
Now, a symbol is removed from \Q{j,k} whenever user~$j$ or user~$k$ receives a linear combination involving one of its elements.  Thus,  \bT{i}, the expected value of $T_i$, is bounded as
\begin{equation}
\label{eq:Ti_Qjk}
	\bar{T}_i \leq \frac{\Q{j,k}^{+}}{1 - \epsilon_j\epsilon_k}.
\end{equation}

Inequality~\eqref{eq:Ti_Qjk} bounds \bT{i} relative to the size of \Q{j,k}.  We must also bound \bT{i} in terms of the cardinality of $Q_i$, since transmissions must stop if $Q_i$ is empty.  This bounding is not as straightforward, however, as symbols are added and removed from $Q_i$ as the algorithm progresses.  For example, $Q_1$ may be empty at a certain time, but later replenished when we send linear combinations of the form $q_2 \oplus q_{1,3}$.  In particular, when this linear combination is received by user~3 and not user~1, we add $q_{1,3}$ to $Q_1$ as user~1 is now the only user in need of it. If $N\bar{T}_2$ linear combinations of this form are sent, we can expect that $N\bar{T}_2\epsilon_1(1 - \epsilon_3)$ symbols are added to $Q_1$.  In general, after $N\bT{j}$ and $N\bT{k}$ linear combinations are sent from $Q_j$, \Q{i,k} and $Q_k$, \Q{i,j} respectively, we can expect that $N\bar{T}_j\epsilon_i(1 - \epsilon_k) + N\bar{T}_k\epsilon_i(1 - \epsilon_j)$ symbols are added to $Q_i$.  Given that $Q_i$ initially has $N\bar{T}_0\epsilon_i(1 - \epsilon_j)(1 - \epsilon_k)$ symbols after $N\bar{T}_0$ uncoded transmissions, we have then that the expected maximum normalized cardinality of $Q_i$ is given by \Qp{i}{\bT{j}}{\bT{k}}  where
\begin{equation}
\label{eq:Qi_plus}
	\Qp{i}{\bT{j}}{\bT{k}} = \bar{T}_0\epsilon_i(1 - \epsilon_j)(1 - \epsilon_k) + \bT{j}\epsilon_i(1 - \epsilon_k) + \bT{k}\epsilon_i(1 - \epsilon_j).
\end{equation}
Since a symbol is removed from $Q_i$ each time user~$i$ successfully receives a channel symbol, we can also write that
\begin{equation}
\label{eq:Ti_Qi}
	\bT{i} \leq \frac{\Qp{i}{\bT{j}}{\bT{k}} }{1 - \epsilon_i}.
\end{equation}

Notice that by the definition of our stopping condition, \bT{i} must actually meet \eqref{eq:Ti_Qjk} or \eqref{eq:Ti_Qi} with equality, since otherwise $Q_i$ and \Q{j,k} would still be non-empty, and we would not have reached the stopping condition.  So we in fact must have that for all $i \in \mathcal{U}$ with $j,k \in \mathcal{U}\setminus \{i\} \ s.t. \ j \neq k$,
\begin{equation}
\label{eq:Ti_min}
	\bT{i} = \min\left(\frac{\Qp{i}{\bT{j}}{\bT{k}} }{1 - \epsilon_i}, \frac{\Q{j,k}^{+}}{1 - \epsilon_j\epsilon_k} \right),
\end{equation}
where $\bT{i} > 0$.
The following theorem proposes that the solution to \eqref{eq:Ti_min} is unique and can be characterized by solving the linear program 

\begin{equation}
\label{eq:LP}
\begin{aligned}
	& \underset{\bT{1}, \bT{2}, \bT{3}}{\text{max}}
	&& \!\! \bT{1} + \bT{2} + \bT{3}\\
	& \text{subject to}
	&& 	\!\! \bT{i} \leq \frac{\Qp{i}{\bT{j}}{\bT{k}} }{1 - \epsilon_i},\\
	&&& \!\!\bar{T}_i \leq \frac{\Q{j,k}^{+}}{1 - \epsilon_j\epsilon_k} \quad \forall i \in \mathcal{U}, j,k \in \mathcal{U}\setminus \{i\}, j \neq k.
\end{aligned}
\end{equation}
\begin{theorem}
\label{thm:LP_instantly_decodable}
	Let $\bT{i}$ be the expected value of the normalized number of analog transmissions that can be sent of the form $q_i \oplus q_{j,k}$, where $i \in \mathcal{U}, j,k \in \mathcal{U}\setminus\{i\}, j\neq k$.  Then $\bT{i}$ is uniquely given by the solution of~\eqref{eq:LP}.
\end{theorem}
\begin{proof}
We proceed toward this end by establishing several lemmas in the appendix in the accompanying supplemental material of this paper and which is also available in the extended version of this paper~\cite{TMK_feedback3_arxiv}.
Lemma~\ifarxiv\ref{lem:LPToMin}\else\lemOPSatisfiesCond{}\fi~\cite{TMK_feedback3_arxiv} first establishes that any optimal solution to~\eqref{eq:LP} also satisfies~\eqref{eq:Ti_min}.  As $\bT{1} = \bT{2} = \bT{3} = 0$ is clearly a feasible solution of~\eqref{eq:LP}, we have that the feasible set of~\eqref{eq:LP} is non-empty.  Therefore, such a solution of~\eqref{eq:LP}, and by extension~\eqref{eq:Ti_min}, indeed exists.  
Conversely, Lemma~\ifarxiv\ref{lem:MinToLP}\else\lemCondSatisfiesOP{}\fi~\cite{TMK_feedback3_arxiv} establishes that any $(\bT{1}, \bT{2}, \bT{3})$ that satisfies~\eqref{eq:Ti_min} is also an optimal solution to~\eqref{eq:LP}.  
%
Finally, Lemma~\ifarxiv\ref{lem:unique}\else\lemOPunique{}\fi~\cite{TMK_feedback3_arxiv} then shows that the optimal solution for~\eqref{eq:LP} is unique.  We conclude that any $(\bT{1}, \bT{2}, \bT{3})$ satisfying~\eqref{eq:Ti_min} is itself unique and characterized by~\eqref{eq:LP}.
\end{proof}
Now, let $\bar{T}^{*} = (\bTs{1}, \bTs{2}, \bTs{3})$ be the optimal solution of~\eqref{eq:LP}, 
and let 
\begin{equation}
\label{eq:t_star}
	t^{*} = \bT{0} + \bTs{1} + \bTs{2} + \bTs{3},
\end{equation}  
where $\bar{T}_{0}$ is given by~\eqref{eq:T0}.
It is important to determine $t^{*}$ since it provides a lower bound for the number of instantly decodable, distortion-innovative transmissions possible.  From the discussion at the beginning of Section~\ref{sec:three_users}, we know that for any latency $w \leq t^{*}$, we can meet the point-to-point outer bound at $w$.  Say that user~$i$'s distortion, $d_i$, necessitates a minimum latency of $w_i(d_i)$ where

\begin{equation}
\label{eq:wioptimal}
	w_i(d_i) = \frac{1 - d_i}{1 - \epsilon_i}.
\end{equation}
Let
%
\begin{eqnarray}
\label{eq:w_minus}
	w^{-}\dvec &=& \min_{i \in \mathcal{U}}  w_i (d_i)\\
\label{eq:w_plus}
	w^{+}\dvec &=& \max_{i \in \mathcal{U}} w_i(d_i).
\end{eqnarray}
It is clear then that $w^{+}$ is an outer bound for our problem, and if $w^{+} \leq t^{*}$, we are optimal for all users.  Notice however, that if $w^{-} \leq t^{*}$, we can also be optimal for all users.  This is because if $w^{-} \leq t^{*}$, one user can be fully satisfied at latency $w^{-}$, and so from thereon, we are left with only two users.  From~\cite{TMKS_TIT20} and the discussion in Section~\ref{subsec:instantly_decodable}, we know that we can remain optimal for the remaining two users, which leads us to the following theorem.

\begin{theorem}
\label{thm:w_minus}
	Given $\dvec \in \mathcal{D}^{3}$, let $t^{*}, w^{-}\dvec$ and $w^{+}\dvec$ be given by~\eqref{eq:t_star}, \eqref{eq:w_minus} and~\eqref{eq:w_plus} respectively.  Then if $w^{-}\dvec \leq t^{*}$, the latency $w^{+}\dvec$ is $(d_1, d_2, d_3)$-achievable.
\end{theorem}
%
If $w^{-} > t^{*}$,  it may still be possible to achieve optimality.  In particular, this may happen if after $Nt^{*}$ transmissions have completed, we are left with non-empty queues $Q_i, i\in \mathcal{U}$. 
We can calculate the expected normalized cardinality of $Q_i$ after $t^{*}$ transmissions, denoted as $| Q_i(t^{*}) |$, by noting that in $NT_i^{*}$ transmissions, a symbol is removed from $Q_i$ with probability $(1 - \epsilon_i)$, and so
\begin{equation}
\label{eq:Q_t_star}
	| Q_i(t^{*}) | = \Qp{i}{\bT{j}^{*}}{\bT{k}^{*}} - \bar{T}_i^{*}(1 - \epsilon_i),
\end{equation}
where \Qp{i}{\bT{i}}{\bT{j}} is given by~\eqref{eq:Qi_plus}.
If $| Q_i(t^{*}) |$ is non-zero for all $i \in \mathcal{U}$, we can continue sending linear combinations of the form $q_1 \oplus q_2 \oplus q_3$ until a user's distortion constraint is met or one of the queues $Q_i$ is exhausted.  If the latter were to happen, we have that user~$i$ has actually reconstructed every source symbol since all queues $Q_U, U \subset \mathcal{U} \ s.t.\ i \in U$ are empty.  Therefore, we are again left with the situation in~\cite{TMKS_TIT20} involving only two users. We conclude that we can send instantly decodable, distortion-innovative symbols until \emph{all} users achieve lossless reconstructions.

\begin{theorem}
\label{thm:all}
	Let $Q^{-} = \min_{i \in \mathcal{U}} |Q_i(t^{*})|$, where $|Q_i(t^{*})|$ is given by~\eqref{eq:Q_t_star}.  If $Q^{-} > 0$, then for \textbf{any} $\dvec \in \mathcal{D}^{3}$, the latency $w^{+}\dvec$ is $(d_1, d_2, d_3)$-achievable.
\end{theorem}

\subsection{Analysis of Chaining Algorithm of Section~\ref{subsubsec:chaining_algorithm}}
\label{subsec:chaining_algorithm_analysis}

In this section, we give a sufficient condition for all users to simultaneously achieve point-to-point optimal performance when the chaining algorithm is invoked.  Our analysis is based on a restricted version of the chaining algorithm previously described.  Specifically, we assume that we do not opportunistically send linear combinations of the forms specified in~\Crefrange{eq:linear_comb_available_a}{eq:linear_comb_available_c} when the opportunities present themselves (see the description of State~4 in Section~\ref{subsubsec:chaining_algorithm}). Thus, we can argue a fortiori that the actual unrestricted chaining algorithm would achieve only better performance.  Before beginning the analysis, we mention that we will rely on many of the results on Markov rewards processes with impulse rewards and absorbing states derived in~\cite{TMK_markov}.  We will cite the specific theorems and corollaries we use from~\cite{TMK_markov} when applicable. 

We begin by deriving the transition matrix for the Markov rewards process.  Consider Table~\ref{tab:state1_transitions_all}.  The table shows all outbound transitions given the channel noise realization, however to construct a transition matrix from this information, we must combine all outbound transitions to the same state.  For example, rows 1 and 5 both show transitions from State~1 to State~5, and so to get $\trans{1}{5}$, the total probability of transitioning from State~1 to~5, we must add the corresponding probabilities under the $\Prob(\zall)$ columns.  This is given by 

\begin{align}
	\trans{1}{5} &= \Prob(\zall = (0, 0, 0)) + \Prob(\zall = (1, 0, 0)) \\
	&= (1 -\epsilon_i)(1 -\epsilon_j)(1 - \epsilon_k) + \epsilon_i(1 -\epsilon_j)(1 - \epsilon_k) \\
	&= (1 -\epsilon_j)(1 - \epsilon_k).
\end{align}
We continue in this manner to find $\trans{i}{j}$ for all $i, j \in \myState \triangleq \{1, 2, \ldots, 6\}$ to populate the transition matrix $\transM$ where the $(i, j)$th entry of $\transM$ is given by $\trans{i}{j}$.

Our analysis also requires the derivation of several rewards matrices.  Consider deriving the rewards matrix $\orhoe$, whose $(i,j)$th element,  $\rholcarg{i}{j}$, gives the expected number of equations (rewards) received by user~$i$ for transitioning from State~$i$ to State~$j$, where $i, j \in \myState$.  Continuing with our previous example, say we would like to derive $\rholcarg{1}{5}$.  Of the two possible paths for an outbound transition from State~1 to~5, only one, when $\zall = (0, 0, 0)$, is associated with a reward in the column of $\rhoc$ in Table~\ref{tab:state1_transitions_all}.  We therefore calculate $\rholcarg{1}{5}$ by weighting the reward with the conditional probability of the transition resulting from the channel noise $\zall = (0, 0, 0)$ given that the transition from State~1 to~5 occurred.  Therefore,

\begin{align}
	\rholcarg{1}{5} &= \Prob(\zall = (0, 0, 0) | \textrm{transition from State~1 to~5}) \cdot 1 \\
	&= \frac{(1 -\epsilon_i)(1 -\epsilon_j)(1 - \epsilon_k)}{(1 -\epsilon_i)(1 -\epsilon_j)(1 - \epsilon_k) + \epsilon_i(1 -\epsilon_j)(1 - \epsilon_k)} \\
	&= (1 -\epsilon_i).
\end{align}
By performing this calculation for all $i, j \in \myState$, we are able to populate the entire rewards matrix $\orhoe$.  Similarly, we can calculate the $|\myState| \times |\myState|$ rewards matrices, $\orhoj$, $\orhok$, $\orhoqi$, $\orhoqj$, $\orhoqk$, and $\orhoqstar$ corresponding to each rewards column in Table~\ref{tab:state1_transitions_all}.

Given the transition matrix, for each reward matrix, we can use the results in \cite[Corollary~\markovcor{}]{TMK_markov} to calculate the expected accumulated reward before absorption \emph{each time} the Markov rewards process is reset and allowed to run until absorption.  To find the \emph{total} expected reward, we must first find a lower bound for $\Mstar$, the number of times the Markov process is \emph{reset} before user~$j$ or~$k$ has met their distortion constraint.  That is, $\Mstar$ is the number of times the Markov rewards process has reached an absorbing state after having been restarted in the initial state.  Therefore, $\Mstar$ also represents the number of chains built by user~$i$ before user~$j$ or~$k$ has satisfied their distortion constraint.

Say that for $r \in \{i, j, k\}$, user~$r$ requires $N(1 - \hat{d}_r)$ symbols at the beginning of the chaining algorithm, where $\hat{d}_r \in (0, \epsilon_r)$.  Furthermore, suppose that of the two users for whom we are targeting point-to-point optimal performance, user~$u$ is \emph{not} the bottleneck user, i.e., 

\begin{align}
\label{eq:u_bottleneck}
	u = \argmin_{r \in \{j, k\}} w_r(\hat{d}_r),
\end{align}
where $w_r(\hat{d}_r)$ is given by~\eqref{eq:wioptimal}.  For $l = 1, 2, \ldots $, let $\oRhojl$ be the expected accumulated reward for user~$u$ during the $l$th time the Markov rewards process has been reset.  Then we can define $\Mstar$ as

\begin{align}
	\Mstar = \min\left\{ m : \sum_{l = 1}^{m} \oRhojl \geq N(1 - \hat{d}_{u}) \right\},
\end{align}
where the $\oRhojl$ are i.i.d.\ and $\mathbb{E}\oRhojl$ can be calculated as in \cite[Corollary~\markovcor{}]{TMK_markov}.

We see that $\Mstar$ is a stopping rule~\cite{Gallager96}.  In order to calculate $\mathbb{E}\Mstar$, we could use the discrete version of the renewal equation~\cite[Chapter~2]{MitovOmey14}.  However, to find a lower bound for $\mathbb{E}\Mstar$, we may simply use Wald's equation\cite{Gallager96}.  Let $\sigma_m = \sum_{l = 1}^{m} \oRhojl$.  Then, by Wald's equation, 

\begin{align}
	\mathbb{E}\Mstar &= \frac{\mathbb{E}\sigma_{\Mstar}}{\mathbb{E}\oRhojl} \\
	\label{eq:EM_lb}
	&\geq \frac{N(1 - \hat{d}_u)}{\mathbb{E}\oRhojl},
\end{align}
where again, $\mathbb{E}\oRhojl$ is calculated as in \cite[Corollary~\markovcor{}]{TMK_markov}.  
Now, let 

\begin{align}
\label{eq:Mstar}
	M^{-} = \left\lfloor{\frac{N(1 - \hat{d}_u)}{\mathbb{E}\oRhojl}}\right\rfloor
\end{align}
be the result of applying the floor function to the right-hand-side of~\eqref{eq:EM_lb}.  We have that $M^{-}$ gives a lower bound for the expected number of times the Markov chain is reset.  If user~$i$, the user building the chains, can meet their distortion constraint within the $M^{-}$ times the Markov rewards process is being run, then all users will be point-to-point optimal.  This is because user~$i$ is able to decode all his required symbols despite the fact that we are targeting optimal performance for the other two users.

Let $\oRhoel$ be the expected number of symbols in the chain that can be decoded in the $l$th run of the Markov rewards process given that the decoding absorbing state, State~6, was reached after having started in the initial state, State~1.  
We mention that we can easily calculate $\oRhoel$ as in~\cite[Theorem~\markovthm{}]{TMK_markov}.  
%
%
Now, let $\overline{\sigma}$ be the expected number of symbols that can be decoded in $M^{-}$ iterations of the Markov rewards process.  By the linearity of the expectation operator we have

\begin{align}
	\overline{\sigma} &= \sum_{l = 1}^{M^{-}} \mathbb{E}\oRhoel \\
	\label{eq:normalized_chain_received}
	&= M^{-} \times \mathbb{E}\oRhoel.
\end{align}
If the right-hand-side of~\eqref{eq:normalized_chain_received} is greater than $N(1 - \hat{d}_i)$, the fraction of symbols user~$i$ requires, then we are point-to-point optimal for all users.  Combining~\eqref{eq:normalized_chain_received} and~\eqref{eq:Mstar}, we see that this happens when 

\begin{align}
\label{eq:chain_optimality}
	\left\lfloor{\frac{N(1 - \hat{d}_u)}{\mathbb{E}\oRhojl}}\right\rfloor \geq \frac{N(1 - \hat{d}_i)}{\mathbb{E}\oRhoel}.
\end{align}

\begin{theorem}
\label{thm:chaining_sufficient}
	Let $u$ be the user satisfying~\eqref{eq:u_bottleneck}.  Then $w^{+}\dvec$ is $\dvec$-achievable if~\eqref{eq:chain_optimality} is satisfied, where $\mathbb{E}\oRhojl$ is calculated from the transition matrix and rewards matrix $\overline{\rho}_{u}$ via \cite[Corollary~\markovcor{}]{TMK_markov}, and $\mathbb{E}\oRhoel$ is calculated from the transition matrix and rewards matrix $\orhoe$ via \cite[Theorem~\markovthm{}]{TMK_markov}.
\end{theorem}

Finally, we again remark that the analysis we have just described is only a sufficient condition for a restricted version of the chaining algorithm in which we do not opportunistically send linear combinations of the form in~\Crefrange{eq:linear_comb_available_a}{eq:linear_comb_available_c}.  Therefore, we expect the unrestricted algorithm to perform better.

\subsection{Operational Significance of Theorem~\ref{thm:chaining_sufficient}}
\label{sec:operational_meaning}

Consider a hypothetical situation in which we have queues $\Q{1, 2}, \Q{1, 3}$ and $\Q{2, 3}$.   We consider quadratic distortions in which for $i \in \{2, 3\}$,  $d_i = \epsilon_i^2$, and we fix $\epsilon_1 = 0.1$, $\epsilon_3 = 0.6$ and vary $\epsilon_2 \in (0.2, 0.6)$.  

We illustrate the chaining algorithm when user~1 is the user who builds chains and we send linear combinations of the symbols in $\Q{1, 2}$ and $\Q{1, 3}$ as if users~2 and~3 were the only users in the network.  In this case, for $i \in \{2, 3\}$, user~$i$'s point-to-point optimal latency is given by $\wi = (1 - \epsilon_i^2)/(1 - \epsilon_i)$, and since $\wi$ is an increasing function of $\epsilon_i \in [0, 1)$, we have that between users~2 and~3, user~2 is \emph{not} the bottleneck user.  That is, in~\eqref{eq:chain_optimality}, user~2 takes the place of user~$u$, and since user~1 is building the chains, user~1 takes the place of user~$i$.

We rearrange~\eqref{eq:chain_optimality} of Theorem~\ref{thm:chaining_sufficient} to find a lower bound for $\hat{d}_i$, and since $\hat{d}_u = \epsilon_u^2$, we  plot this lower bound as a function of $\epsilon_u$ in Figure~\ref{fig:operational_meaning}.  From this figure, we can read which values of $\hat{d}_i$ would yield optimal performance for a given value of $\epsilon_u$ by considering all values of $\hat{d}_u$ above the lower bound.  Recall that user~$i$ is able to decode symbols in the chain only if there have been two consecutive transmissions for which user~$i$ is the only user to have received the transmission.  We see that the probability of this event increases as $\epsilon_u$ increases in Figure~\ref{fig:operational_meaning}.  Therefore, user~$i$ is able to achieve lower distortions as $\epsilon_u$ increases.  

Finally, we again mention that Theorem~\ref{thm:chaining_sufficient} merely gives a conservative \emph{sufficient} condition for optimality and ignores other network coding opportunities in its analysis.  Therefore, it is possible that point-to-point optimal performance can still be met if user~$i$ has a distortion constraint below the lower bound of Figure~\ref{fig:operational_meaning}.

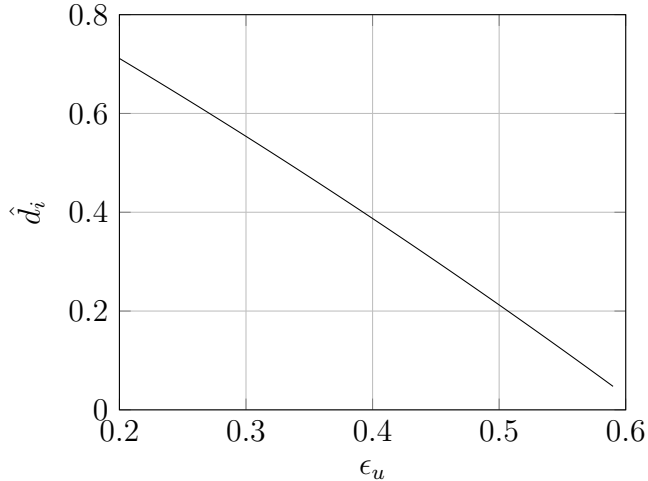
\begin{figure}
	\centering
	\setlength\figurewidth{2.65in} 
	\setlength\figureheight{2.07in} 
%
%
%
%
\begin{tikzpicture}

\begin{axis}[%
view={0}{90},
width=\figurewidth,
height=\figureheight,
scale only axis,
xmin=0.2, xmax=0.6,
xlabel={$\epsilon_{u}$},
xmajorgrids,
ymin=0, ymax=0.8,
ylabel={$\hat{d}_i$},
ymajorgrids]
\addplot [
solid
]
coordinates{
 (0.2,0.711373104662686)(0.21,0.695995143392607)(0.22,0.680529923158764)(0.23,0.664977674184357)(0.24,0.649338602849747)(0.25,0.633612892531508)(0.26,0.617800704409526)(0.27,0.601902178243395)(0.28,0.585917433119295)(0.29,0.569846568168493)(0.3,0.55368966325855)(0.31,0.537446779658295)(0.32,0.521117960677552)(0.33,0.504703232282582)(0.34,0.488202603688164)(0.35,0.471616067927191)(0.36,0.454943602398622)(0.37,0.438185169394609)(0.38,0.421340716607562)(0.39,0.404410177617902)(0.4,0.387393472363215)(0.41,0.370290507589491)(0.42,0.353101177285094)(0.43,0.335825363098112)(0.44,0.318462934737666)(0.45,0.301013750359774)(0.46,0.283477656938324)(0.47,0.265854490621674)(0.48,0.248144077075411)(0.49,0.230346231811744)(0.5,0.212460760506007)(0.51,0.194487459300731)(0.52,0.176426115097705)(0.53,0.158276505838452)(0.54,0.140038400773524)(0.55,0.12171156072098)(0.56,0.103295738314447)(0.57,0.0847906782410863)(0.58,0.0661961174698284)(0.59,0.0475117854701842) 
};

\end{axis}
\end{tikzpicture}
	\caption{The lower bound from Theorem~\ref{thm:chaining_sufficient} that delineates a conservative boundary of distortion values for which the minmax optimal latency can be achieved.  }
	\label{fig:operational_meaning}
\end{figure}

\section{Simulations}
\label{sec:simulations}

We demonstrate the performance of the algorithms in Section~\ref{sec:three_users} with simulations.  We consider distortions for which $d_i = \epsilon_i^{2}$ for all $i \in \mathcal{U}$.  In Fig.~\ref{fig:t_star}, we choose a blocklength of $N=10^7$, fix $\epsilon_1 = 0.3$, $\epsilon_2 = 0.4$ and vary $\epsilon_3$ on the $x$-axis while plotting the total number of channel symbols sent per source symbol on the $y$-axis for different coding schemes.  
%
\begin{figure}
	\centering
	\setlength\figurewidth{2.65in} 
\setlength\figureheight{2.07in} 
%
%
%
%
\begin{tikzpicture}

\scriptsize

\begin{axis}[%
view={0}{90},
width=\figurewidth,
height=\figureheight,
scale only axis,
xmin=0.85, xmax=0.95,
xlabel={\scriptsize $\epsilon{}_\text{3}$},
xmajorgrids,
grid style={dotted, thick},
ymin=1.15, ymax=1.65,
ylabel={\scriptsize Channel Symbols Sent per Source Symbol},
ymajorgrids,
grid style={dotted, thick},
legend style={nodes=right}]
\addplot [
color=black,
solid
]
coordinates{
 (0.85,1.63951579495798)(0.86,1.61861464844961)(0.87,1.59617446412151)(0.88,1.57517375584416)(0.89,1.55431985096308)(0.9,1.53333057539682)(0.91,1.51154970934066)(0.92,1.49059085010352)(0.93,1.46928116062468)(0.94,1.44938128693009)(0.95,1.42717339035088) 
};

\addlegendentry{\scriptsize Segmentation-based};

\addplot [
color=black,
only marks,
mark=o,
mark options={solid}
]
coordinates{
 (0.85,1.2855624)(0.86,1.2751091)(0.87,1.2640831)(0.88,1.2538016)(0.89,1.2436843)(0.9,1.2336025)(0.91,1.2232663)(0.92,1.2133772)(0.93,1.2034194)(0.94,1.1941516)(0.95,1.1839635) 
};

\addlegendentry{\scriptsize Simulations};

\addplot [
color=black,
dashed
]
coordinates{
 (0.85,1.28528689843124)(0.86,1.2747660790011)(0.87,1.26433699565063)(0.88,1.25399790893119)(0.89,1.24374711701383)(0.9,1.23358295457952)(0.91,1.22350379167497)(0.92,1.21350803280185)(0.93,1.2035941147971)(0.94,1.19376051087361)(0.95,1.18400571842039) 
};

\addlegendentry{\scriptsize Solution to~\eqref{eq:LP}};

\end{axis}
\end{tikzpicture}
	\caption{The normalized number of channel symbols sent per source symbol.  We fix $\epsilon_1=0.3$, $\epsilon_2=0.4$ and vary $\epsilon_3$ on the x-axis.  We show the number of uncoded transmissions sent via~\eqref{eq:LP}, alongside what is obtained from simulations  for a blocklength of $N=10^7$.  Finally, we also plot the latency required to achieve the equivalent distortion values \emph{without} feedback based on a segmentation-based coding scheme.}
	\label{fig:t_star}
\end{figure}
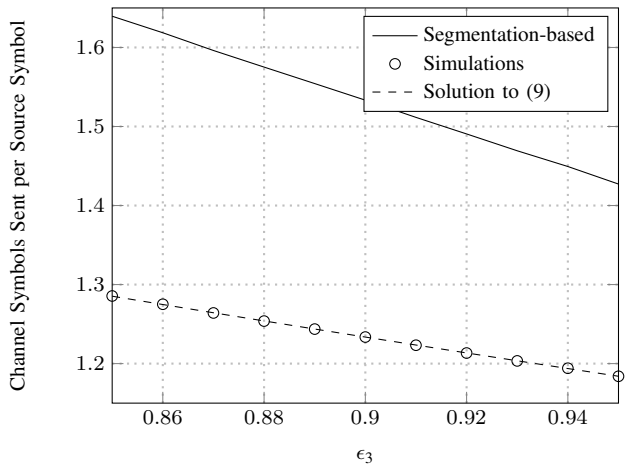
The first coding scheme plotted is based on simulations and plots the total number of instantly decodable, distortion-innovative transmissions sent.  Alongside this curve, we plot the number of possible innovative transmissions suggested by the solution of~\eqref{eq:LP}.  We observe a close resemblance in this plot and the empirical simulation curve. Finally, we know that although each user may not have their final distortion constraint met after $Nt^{*}$ transmissions, the provisional distortion they \emph{do} achieve after $Nt^{*}$ transmissions is optimal.  Say instead, that we were given these provisional distortion values from the onset and asked what latency would be required to achieve these distortions if feedback were not available.  The final plot shows this required latency for the segmentation-based coding scheme of~\cite{LTKS_ISIT14}, which does not incorporate feedback. The gap between these curves is indicative of the benefit that feedback provides.

In practice, the values of $\epsilon_i$ are much lower than what we have chosen.  The values for $\epsilon_3$, for example, were deliberately chosen to be high ($\epsilon_3 \geq 0.85$) as we have found that for values even as high as $\epsilon_3 = 0.8$, we achieve point-to-point optimality for all users (see Theorems~\ref{thm:w_minus} and~\ref{thm:all}).  If we increase $\epsilon_3$ even higher however, we observe a situation where many symbols destined to user~3 are erased, and so when the stopping condition of the algorithm in Section~\ref{subsec:instantly_decodable} is reached, we are left with queues $Q_3$, \Q{1,3} and \Q{2,3}.
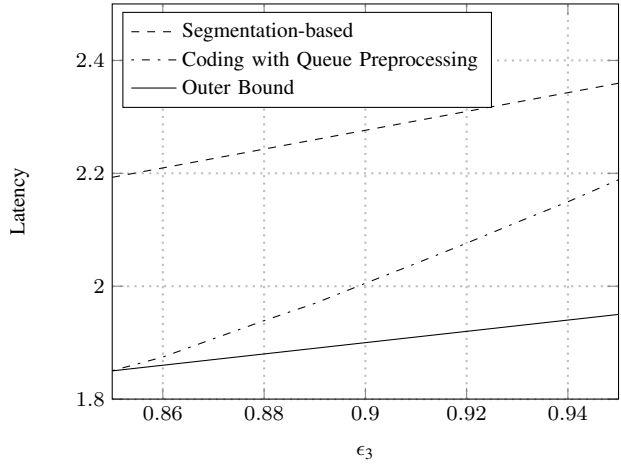
\begin{figure}
	\centering
	\setlength\figurewidth{2.65in} 
	\setlength\figureheight{2.07in} 
%
%
%
%
\begin{tikzpicture}

\scriptsize

\begin{axis}[%
view={0}{90},
width=\figurewidth,
height=\figureheight,
scale only axis,
xmin=0.85, xmax=0.95,
xlabel={$\epsilon{}_\text{3}$},
xmajorgrids,
grid style={dotted, thick},
ymin=1.8, ymax=2.5,
ylabel={Latency},
ymajorgrids,
grid style={dotted, thick},
legend style={nodes=right, at={(0.02,0.98)}, anchor=north west}]
\addplot [
color=black,
dashed
]
coordinates{
 (0.85,2.19285714285714)(0.86,2.20952380952381)(0.87,2.22619047619048)(0.88,2.24285714285714)(0.89,2.25952380952381)(0.9,2.27619047619048)(0.91,2.29285714285714)(0.92,2.30952380952381)(0.93,2.32619047619048)(0.94,2.34285714285714)(0.95,2.35952380952381) 
};

\addlegendentry{\scriptsize Segmentation-based};

\addplot [
color=black,
dash pattern=on 1pt off 3pt on 3pt off 3pt
]
coordinates{
 (0.85,1.8503066)(0.86,1.874513)(0.87,1.9069384)(0.88,1.9389846)(0.89,1.9695207)(0.9,2.0051401)(0.91,2.0403425)(0.92,2.0762682)(0.93,2.1133212)(0.94,2.1495253)(0.95,2.1885717) 
};

\addlegendentry{\scriptsize Coding with Queue Preprocessing};

\addplot [
color=black,
solid
]
coordinates{
 (0.85,1.85)(0.86,1.86)(0.87,1.87)(0.88,1.88)(0.89,1.89)(0.9,1.9)(0.91,1.91)(0.92,1.92)(0.93,1.93)(0.94,1.94)(0.95,1.95) 
};

\addlegendentry{\scriptsize  Outer Bound};

\end{axis}
\end{tikzpicture}
	\caption{The latency when we are forced to invoke the algorithms of Section~\ref{subsec:non_instant_coding}.  We fix $d_i = \epsilon_{i}^{2}$, $\epsilon_1 = 0.3$, $\epsilon_2 = 0.4$ and vary $\epsilon_3$ on the x-axis.  }
	\label{fig:t_total}
\end{figure}
When this occurs, we have not yet satisfied all users, and so we resort to the coding schemes proposed in Section~\ref{subsec:non_instant_coding}.  Fig.~\ref{fig:t_total} shows a plot where each point required the invocation of the queue preprocessing scheme for channel coding in Section~\ref{subsec:channel_coding}.  It plots the \emph{overall} latency required to achieve distortions $d_i = \epsilon_i^{2}$ for the values of $\epsilon_i$ given earlier.  Alternatively, if the chaining algorithm of Section~\ref{subsubsec:chaining_algorithm} is used, we find that the latency coincides with the outer bound.  Again, the segmentation-based scheme is provided as a benchmark along with the outer bound $w^{+}\dvec$.

\ifarxiv
\ifarxiv
\appendices
\fi
\section{Proof of Supporting Lemmas of Theorem~\ref{thm:LP_instantly_decodable}}


\begin{myLemma}
\label{lem:LPToMin}
	Let $(\bT{1}, \bT{2}, \bT{3})$ be the optimal solution of the linear program in~\eqref{eq:LP}.  Then $(\bT{1}, \bT{2}, \bT{3})$ satisfies~\eqref{eq:Ti_min}.
\end{myLemma}
\begin{proof}
To show that the optimal solution of~\eqref{eq:LP} must satisfy \eqref{eq:Ti_min}, we note that in order to maximize the objective function in \eqref{eq:LP}, one of the inequality constraints for \bT{i} must be met with equality.  Otherwise, if \bT{i}, \bT{j}, \bT{k} is a purported optimal solution where \bT{i} does not meet an inequality constraint with equality, we can find some $\delta > 0$ such that $\bT{i} + \delta$, \bT{j}, \bT{k} is also feasible and gives a strictly larger objective function value.
\end{proof}


\begin{myLemma}
\label{lem:MinToLP}
	Let $(\bT{1}, \bT{2}, \bT{3})$ satisfy~\eqref{eq:Ti_min}, where $\bT{i} > 0$ for all $i \in \mathcal{U}$.  Then $(\bT{1}, \bT{2}, \bT{3})$ is an optimal solution of~\eqref{eq:LP}.
\end{myLemma}
\begin{proof}
	We show that $(\bT{1}, \bT{2}, \bT{3})$ is a local maximum of the linear program in~\eqref{eq:LP}.  Consider a neighbouring  element in the feasible region of~\eqref{eq:LP} of the form $(\bT{1} + \Delta_1, \bT{2} + \Delta_2, \bT{3} + \Delta_3)$.  We show that the objective function evaluated at this new element is not greater than the objective function evaluated at $(\bT{1}, \bT{2}, \bT{3})$.  We do this by showing that the linear program

	
\begin{equation}
\label{eq:local}
\begin{aligned}
	& \underset{\Delta_{1}, \Delta_{2}, \Delta_{3}}{\text{max}}
	&& \!\! \Delta_{1} + \Delta_{2} + \Delta_{3}\\
	& \text{subject to}
	&& 	\!\! \bT{i} + \Delta_i \leq \frac{\Qp{i}{\bT{j} + \Delta_j}{\bT{k} + \Delta_k} }{1 - \epsilon_i},\\
	&&& \!\!\bar{T}_i + \Delta_i \leq \frac{\Q{j,k}^{+}}{1 - \epsilon_j\epsilon_k} \quad \forall i \in \mathcal{U}, j,k \in \mathcal{U}\setminus \{i\} , j \neq k.
\end{aligned}
\end{equation}
has an optimal value no greater than zero.  

By assumption, $(\bT{1}, \bT{2}, \bT{3})$ satisfies~\eqref{eq:Ti_min}, and so for all $i \in \mathcal{U}$, $\bT{i}$ meets one of the inequality constraints in~\eqref{eq:LP} with equality. Let $\mathcal{T}(\bT{1}, \bT{2}, \bT{3})$ be the set of $i$ for which \bT{i} meets the first inequality constraint with equality, i.e., 
\begin{flalign}
\label{eq:Tset}
	&\mathcal{T}(\bT{1}, \bT{2}, \bT{3}) = \\ \nonumber
	& \left\{i \in \mathcal{U} \mid \bT{i} = \frac{\Qp{i}{\bT{j}}{\bT{k}} }{1 - \epsilon_i}, \bT{i} > 0, j, k \in \mathcal{U} \setminus \{i\}, j \neq k\right\}.
\end{flalign}
When the context is clear, we will at times use $\mathcal{T}$ to refer to $\mathcal{T}(\bT{1}, \bT{2}, \bT{3})$.  Furthermore, for $i \in \mathcal{T}$, we will at times find it convenient to emphasize the linear dependence of \bT{i} on \bT{j}, and \bT{k} by writing (c.f.~\eqref{eq:Qi_plus})
\begin{equation}
\label{eq:linear_Qi_plus}
	\frac{Q_i^{+}(\bT{j}, \bT{k})}{1 - \epsilon_i} = k_i + a_{ik}\bT{j} + a_{ij} \bT{k},
\end{equation}
where 
\begin{equation}
\label{eq:a_ik}
	a_{ik} = \frac{\epsilon_i}{(1 - \epsilon_i)}(1 - \epsilon_k) > 0,
\end{equation}
\begin{equation}
\label{eq:k_i}
	k_i = \frac{\bT{0}\epsilon_i(1 - \epsilon_j)(1 - \epsilon_k)}{1 - \epsilon_i} > 0,
\end{equation}
and \bT{0} is given by~\eqref{eq:T0}.

Now, we have that if $i \notin \mathcal{T}$, then $\bT{i} = \Q{j,k}^{+}/(1 - \epsilon_j\epsilon_k)$.  From the second inequality of~\eqref{eq:local}, we have then that for these values of $i$, $\Delta_i \leq 0$.  Notice  however, that an optimal solution of~\eqref{eq:local} must have all of its components non-negative.  Otherwise, a larger-valued objective function could be obtained by setting any negative component to zero.  
We can therefore upper bound the optimal value for the linear program in~\eqref{eq:local} with the negated optimal value for the \emph{relaxed} linear program

%
\begin{equation}
\label{eq:delta}
\begin{aligned}
	& \underset{\Delta_{1}, \Delta_{2}, \Delta_{3}}{\text{min}}
	&& \!\! -(\Delta_{1} + \Delta_{2} + \Delta_{3})\\
	& \text{subject to}
	&& 	\!\! \Delta_i \leq a_{ik}\Delta_{j} + a_{ij}\Delta_{k} \\
	&&& \quad \forall i \in \mathcal{T}, j,k \in \mathcal{U} \setminus\{i\}, j\neq k\\
	&&& \Delta_j = 0, \quad \forall j \notin \mathcal{T}.
\end{aligned}
\end{equation}
where we have used~\eqref{eq:Tset} and \eqref{eq:a_ik} to simplify the constraints of~\eqref{eq:local}.

We now consider several cases.  Notice first that if $0 \leq |\mathcal{T}| \leq 1$, the constraints of~\eqref{eq:delta} clearly show that the optimal value of~\eqref{eq:delta} is zero. Alternatively, if $2 \leq |\mathcal{T}| \leq 3$, we consider the Lagrange dual function, $g(\lambda)$, given by~\cite[Section~5.2.1]{boyd2004convex}
\begin{equation}
g(\lambda) = 
	\begin{cases}
   		0 & \text{if } A^T \lambda = -c, \\
   		-\infty & \text{else,}
	\end{cases}
\end{equation}  
where we have assumed that the objective function and inequality constraints of~\eqref{eq:delta} have been written as $c^T\Delta$ and $A\Delta \leq 0$ respectively for some $\Delta$, $A$ and $c$ that will be subsequently defined depending on $|\mathcal{T}|$.  For any $\lambda \succeq 0$, $g(\lambda)$ gives a lower bound on the optimal value of~\eqref{eq:delta}.  If we can therefore find a $\lambda \succeq 0$ such that 
\begin{equation}
\label{eq:lambda}
A^T \lambda = -c, 
\end{equation}
we will have thus shown that the optimal value of~\eqref{eq:local} is no greater than zero, as was our original goal.

\begin{enumerate}

	\item Consider now, if $|\mathcal{T}| = 2$.  Let $i, j$  and $k$ be distinct elements in $\mathcal{U}$, where we assume without loss of generality that $k$ is the only element not in~$\mathcal{T}$.
	In this case, we have
	\begin{equation}
		A =  A_2 \triangleq\begin{bmatrix}
			       	1 & -a_{ik}           \\
				-a_{jk} & 1  
			\end{bmatrix},
	\end{equation}
	\begin{equation}
	\label{eq:c2}
		c = c_2 \triangleq \begin{bmatrix}
				-1 \\
				-1 
			\end{bmatrix}, \qquad
		\Delta = \hat{\Delta} \triangleq \begin{bmatrix}
				\Delta_i \\
				\Delta_j 
			\end{bmatrix},
	\end{equation}
	where we remind the reader that we have assumed that the constraints in~\eqref{eq:delta} are written as $A\Delta \leq 0$.  
	Since the off-diagonal entries of $A_2$ are negative by~\eqref{eq:a_ik}, we have that $A_2$ is a $Z$-matrix~\cite{P77}.  Furthermore, since $|\mathcal{T}| = 2$ by assumption, we further have that there exists an $x \triangleq [\bT{i}, \bT{j}] \succeq 0$ such that $A_2x = k \succ 0$, where $k = [k_i + a_{ij}\bT{k}, k_j + a_{ji}\bT{k}]$ and $\bT{k} = \Q{i, j}^{+}/(1 - \epsilon_i\epsilon_j)$, (see~\eqref{eq:Tset}, \eqref{eq:linear_Qi_plus} and \eqref{eq:k_i}).  Therefore, by Condition $K_{34}$ of~\cite{P77}, $A_2$ is also a non-singular $M$-matrix.  
	By Condition~$F_{15}$ of~\cite{P77}, we thus infer that the inverse matrix $A_{2}^{-1}$ has all positive elements.  Finally, we conclude from~\eqref{eq:c2} that indeed $[\lambda_1, \lambda_2]^T = -(A_2^T)^{-1} c_2 \succeq 0$ since all elements involved in the matrix multiplication are positive.

	
	\item If $|\mathcal{T}| = 3$, we have that
	\begin{equation}
	\label{eq:A3}
		A = A_3 \triangleq \begin{bmatrix}
			       	1 & -a_{ik} & -a_{ij}  \\
				-a_{jk} & 1 & -a_{ji}  \\
				-a_{kj} & -a_{ki} & 1 
			\end{bmatrix},
	\end{equation}
	\begin{equation}
	\label{eq:c3}
		c = c_3 \triangleq \begin{bmatrix}
				-1 \\
				-1 \\
				-1
			\end{bmatrix}, \qquad
		\Delta = \mathring{\Delta} \triangleq \begin{bmatrix}
				\Delta_i \\
				\Delta_j \\
				\Delta_k
			\end{bmatrix}.
	\end{equation}
	
	Similar to the case in which $|\mathcal{T}| = 2$, we can again argue that $A_{3}$ is a non-singular $M$-matrix and so it is inverse-positive~\cite{P77}.  Consequently, we again have that indeed $[\lambda_1, \lambda_2, \lambda_3]^T = -(A_3^{T})^{-1} c_3 \succeq 0$.
	
%
\end{enumerate}


\end{proof}


\begin{myLemma}
\label{lem:unique}
	Let $(\bT{1}, \bT{2}, \bT{3})$ be the optimal solution of the linear program in~\eqref{eq:LP}.  Then $(\bT{1}, \bT{2}, \bT{3})$  is unique.
\end{myLemma}

\begin{proof}
We show that the solution to~\eqref{eq:LP} is unique via Theorem~1 of~\cite{OLM}.  
Intuitively, its reasoning is that to maximize the objective function of a linear program, we follow the objective function's gradient vector until we reach the boundary of the feasible region.  If this stopping occurs at a face of the region rather than a single point, we have a non-unique solution.  In this case, following the gradient vector after it has undergone a small perturbation in its direction will lead us to a different solution on the face of the boundary region.  On the other hand, if any small perturbation in the gradient vector's direction leads us to the same boundary point as the unperturbed gradient vector, we know we have a unique solution.  This is stated in the following fact.

\begin{fact}[{\hspace{1sp}\cite[Theorem 1]{OLM}}]
\label{fact:unique}
	A solution $\bar{x}$ to the linear programming problem
	\begin{equation}
	\label{eq:LP_fact}
\begin{aligned}
	& \underset{x}{\text{max}}
	&& c^{T}x\\
	& \text{subject to}
	&& 	Ax \leq b\\
\end{aligned}
\end{equation}
is unique iff for all $q$, there exists a $\delta >0$ s.t.\ $\bar{x}$ is still a solution when the objective function is replaced by $(c + \delta q)^{T}x$.
\end{fact}

For our purposes, we have that $c = [1, 1, 1]^{T}$ in Fact~\ref{fact:unique}. Furthermore, for any given $q$, we choose $\delta$ large enough so that all coefficients of the newly-perturbed objective function remain positive.  

Let $\bar{T}^{*} = (\bT{1}^{*}, \bT{2}^{*}, \bT{3}^{*})$ be the optimal solution of~\eqref{eq:LP}, and let $\bar{T} = (\bT{1}, \bT{2}, \bT{3})$ be another element in the feasible region of~\eqref{eq:LP}.  We will show that $\bar{T}^{*}$ remains the optimal solution of~\eqref{eq:LP} when the objective function is perturbed in a way that maintains the positivity of its coefficients by showing that  
\begin{equation}
\label{eq:unique_positive}
	\bT{i}^{*} - \bT{i} \geq 0, \qquad \textrm{for all } i \in \mathcal{U}.
\end{equation}
Assume, by way of contradiction, that~\eqref{eq:unique_positive} does not hold, i.e., there exists a non-empty set $\mathcal{V} \subseteq \mathcal{U}$ such that for all $j \in \mathcal{V}$, $\bT{j}^{*} - \bT{j} < 0$.  We will show that the existence of such a set violates the optimality assumption of $\bar{T}^{*}$.  In particular, we construct a feasible solution of~\eqref{eq:LP} with a strictly larger value for the objective function than when it is evaluated at $\bar{T}^{*}$.  This new solution retains all $\bT{i}^{*}$ for $i \in \mathcal{U} \setminus \mathcal{V}$ and replaces all $\bT{j}^{*}$ with $\bT{j}$ for all $j \in \mathcal{V}$.  

For all $i \in \mathcal{U} \setminus \mathcal{V}$, we now show the feasibility of $\bT{i}^{*}$ within the newly-constructed solution by writing
\begin{equation}
\label{eq:still_feasible1}
	\bT{i}^{*} \leq \frac{Q_{\{j, k\}}^{+}}{1 - \epsilon_j\epsilon_k},
\end{equation}
and 
\setcounter{cnt}{1}
\begin{eqnarray}
	\bT{i}^{*} 
		&\stackrel{(\alph{cnt})}{\leq}& k_i + a_{ik} \bT{j}^{*} + a_{ij}\bT{k}^{*} \\
		\addtocounter{cnt}{1}
		&=& k_i + \sum_{\substack{u \in \mathcal{U} \setminus \mathcal{V}, \\ v \in \mathcal{U} \setminus \{i, u\}}} a_{iv} \bT{u}^{*} + \sum_{\substack{v \in \mathcal{V} \\ u \in \mathcal{U} \setminus \{i, v\}}} a_{iu} \bT{v}^{*} \\
		\label{eq:still_feasible}
		&\stackrel{(\alph{cnt})}{<}& k_i + \sum_{\substack{u \in \mathcal{U} \setminus \mathcal{V}, \\ v \in \mathcal{U} \setminus \{i, u\}}}a_{iv} \bT{u}^{*} + \sum_{\substack{v \in \mathcal{V} \\ u \in \mathcal{U} \setminus \{i, v\}}} a_{iu} \bT{v},
\end{eqnarray}
where 
\begin{enumerate}[(a)]
	\item and \eqref{eq:still_feasible1} follow from~\eqref{eq:linear_Qi_plus} and the feasibility of $\bar{T}^{*}$
	\item follows from~\eqref{eq:a_ik} and the definition of $\mathcal{V}$.
\end{enumerate}

Hence, \eqref{eq:still_feasible} shows that $\bT{i}^{*}$ remains feasible within the newly-constructed solution.  Similarly, we can show that $\bT{j}$ remains feasible within the newly-constructed solution for all $j \in \mathcal{V}$.

%

\end{proof}

\fi

\ifCLASSOPTIONcaptionsoff
  \newpage
\fi



%
%
%
%

\vspace{-1em}

\bibliographystyle{IEEEtran}
\bibliography{IEEEabrv,itw2013emina2}

\vfill


\end{document}